\documentclass[cupthm,crop,info]{CUP-JNL-ETS}%

\usepackage{graphicx}
\usepackage{multicol,multirow}
\usepackage{amsmath,amssymb,amsfonts}
\usepackage{mathrsfs}
\usepackage{rotating}
\usepackage{appendix}
\usepackage[numbers]{natbib}
\usepackage{amsthm}
\usepackage{lipsum}

\usepackage{url}
\usepackage{float}
\usepackage{caption}
\usepackage{standalone}
\usepackage{tikz}
\usetikzlibrary{arrows,decorations.pathreplacing,decorations.markings,decorations.pathmorphing,backgrounds,positioning,fit,petri,knots,arrows,patterns,shapes.multipart,shapes.arrows}
\usepackage{pgfplots}
\makeatletter

\newskip\@bigflushglue \@bigflushglue = -100pt plus 1fil

\def\bigcentering{\let\\\@centercr\rightskip\@bigflushglue%
\leftskip\@bigflushglue
\parindent\z@\parfillskip\z@skip}

\makeatother

\usepackage{color}
\definecolor{rouge}{RGB}{255,77,77}
\definecolor{vert}{RGB}{0,178,102}
\definecolor{jaune}{RGB}{255,255,0}
\definecolor{violet}{RGB}{208,32,144}
\definecolor{orange}{RGB}{255,140,0}
\definecolor{bleu}{RGB}{0,0,205}

\pgfplotsset{compat=1.17}

\theoremstyle{cupplain}
\newtheorem{theorem}{Theorem}[section]
\newtheorem{lemma}[theorem]{Lemma}
\newtheorem{prop}[theorem]{Proposition}
\newtheorem{cor}[theorem]{Corollary}
\newtheorem{conjecture}{Conjecture}

\newtheorem{innercustomthm}{Theorem}
\newenvironment{mythm}[1]
  {\renewcommand\theinnercustomthm{#1}\innercustomthm}
  {\endinnercustomthm}
  
\newtheorem{innercustomcor}{Corollary}
\newenvironment{mycor}[1]
  {\renewcommand\theinnercustomcor{#1}\innercustomcor}
  {\endinnercustomcor}  

\theoremstyle{cupdefinition}
\newtheorem{defn}{Definition}[section]
\theoremstyle{cupremark}
\newtheorem{rem}[theorem]{Remark}

\theoremstyle{cupproof}
\newtheorem{proof}{Proof}

\numberwithin{equation}{section}

\newcommand{\F}{\mathbb{F}}
\newcommand{\G}{\mathcal{G}}

\newcommand{\Z}{\mathbb{Z}}
\newcommand{\R}{\mathbb{R}}
\newcommand{\N}{\mathbb{N}}
\newcommand{\ter}{\mathfrak{t}}
\newcommand{\init}{\mathfrak{i}}
\newcommand{\Conv}{\mathop{\scalebox{2}{\raisebox{-0.2ex}{$\ast$}}}}
\newcommand{\argmax}{\textnormal{argmax}}
\newcommand{\dist}{\textnormal{dist}}

\begin{document}

\begin{Frontmatter}

\title{Strongly Aperiodic SFTs on Generalized Baumslag-Solitar groups}

\author[1]{\gname{Nathalie} \sname{Aubrun}}
\author[1]{\gname{Nicol\'as} \sname{Bitar}}
\author[2]{\gname{Sacha} \sname{Huriot-Tattegrain}}

\address[1]{\orgdiv{CNRS, Université Paris-Saclay}, \orgname{LISN}, \orgaddress{\city{Orsay}, \postcode{91400}, \country{France}}\\ (\email{nathalie.aubrun@lisn.fr}, \email{nicolas.bitar@lisn.upsaclay.fr})}
\address[2]{\orgname{École Normale Supérieure Paris-Saclay}, \orgaddress{\city{Gif-sur-Yvette}, \postcode{91190},  \country{France}}\\ (\email{sachahuriot@gmail.com})}


\maketitle

\authormark{N. Aubrun, N. Bitar and S. Huriot-Tattegrain}
\titlemark{Strongly Aperiodic SFTs on Generalized Baumslag-Solitar groups}

\begin{abstract}
We look at constructions of aperiodic SFTs on fundamental groups of graph of groups. In particular we prove that all generalized Baumslag-Solitar groups (GBS) admit a strongly aperiodic SFT. Our proof is based on a structural theorem by Whyte and on two constructions of strongly aperiodic SFTs on $\mathbb{F}_n\times \mathbb{Z}$ and $BS(m,n)$ of our own. Our two constructions rely on a path-folding technique that lifts an SFT on $\mathbb{Z}^2$ inside an SFT on $\mathbb{F}_n\times \mathbb{Z}$ or an SFT on the hyperbolic plane inside an SFT on $BS(m,n)$. In the case of $\mathbb{F}_n\times \mathbb{Z}$ the path folding technique also preserves minimality, so that we get minimal strongly aperiodic SFTs on unimodular GBS groups.
\end{abstract}



\end{Frontmatter}

\section*{Introduction}

 Subshifts are closed and $G$-invariant subsets of $A^G$, with $G$ a finitely generated group and $A$ a finite alphabet. These are fundamental and ubiquitous examples, they can be viewed both as dynamical systems and as computational models. Interestingly, every expansive $G$-action can be encoded inside a subshift that shares the same dynamical properties: the complexity is transferred from the action to the phase space --the subshift in this case. This dynamical definition of subshifts has a combinatorial equivalent: subshifts are also subsets of $A^G$ that respect local rules given by forbidden patterns (forbidden finite configurations). Given a set of forbidden patterns $\mathcal{F}$, the subshift $X_{\mathcal{F}}$ it defines is the set of configurations $x\in A^G$ that avoid all patterns from $\mathcal{F}$. Any set of patterns --not necessarily finite and even not necessarily computable-- defines a subshift, while the same subshift may be defined by several sets of forbidden patterns. Subshifts of finite type (SFT for short) --those for which the set of forbidden patterns can be chosen finite-- form an interesting class of subshifts, since they can model real life phenomenon described by local interactions, and also define a model of computation whose computational power depends on the group $G$.

\medskip

Among questions related to SFTs, an open problem is to characterize groups~$G$ that admit a strongly aperiodic SFT, \emph{i.e.} an SFT such that all its configurations have a trivial stabilizer. 
The interest of aperiodic SFTs is two-fold. On the one hand, they evidence how finite local configurations can generate complex global behaviour. For instance, constructions of strongly aperiodic SFTs are often key elements of proofs of undecidability of the Domino Problem. This suggests a possible connection between the two phenomena; indeed proofs of the undecidability of the original Domino Problem made extensive use of aperiodicity \cite{Jeandel_Vanier_2020}. On the other hand, Gromov noted in \cite{gromov1987hyperbolic} that the existence of aperiodic tileset has close affinity with the question of whether or not the fundamental group of CAT(0) spaces contain $\Z^2$. The question of the existence of aperiodic subshifts in general --non SFT-- has been answered positively~\cite{GaoJacksonSeward2009,AubrunBarbieriThomasse2018}. Free groups and more generally groups with at least two ends cannot possess a strongly aperiodic SFT~\cite{cohen2017large} and a finitely generated and recursively presented group with an aperiodic SFT necessarily has a decidable word problem~\cite{Jeandel2015aperiodic}. Groups that are known to admit a strongly aperiodic SFT are $\Z^2$~\cite{Robinson1971} and $\Z^d$ for $d>2$~\cite{CulikKari1996}, fundamental groups of oriented surfaces~\cite{CohenGoodmanStrauss2017}, 1-ended hyperbolic groups~\cite{CohenGSRieck2017}, discrete Heisenberg group~\cite{SahinSchraudnerUgarcovici2020} and more generally groups that can be written as a semi-direct product $G = \Z^2 \rtimes_{\phi} H$--provided $G$ has decidable word problem~\cite{BarbieriSablik2019}--, self-simulable groups with decidable word problem~\cite{BSS2021} --which include braid groups and some RAAGs-- and residually finite Baumslag-Solitar groups~\cite{EsnayMoutot2020}. Combining results in~\cite{Jeandel2015aperiodic} and~\cite{cohen2017large} one may conjecture that finitely generated groups admitting a strongly aperiodic SFT are exactly 1-ended groups with decidable word problem. 

\begin{conjecture}\label{conjecture:aperiodic_SFT}
A finitely generated group admits a strongly aperiodic SFT if and only if it is 1-ended and has decidable word problem.
\end{conjecture}

There are several techniques to produce aperiodicity inside tilings, but in the particular setting of $\Z^2$-SFTs, two of them stand out. The first goes back to Robinson~\cite{Robinson1971}: the aperiodicity of Robinson's SFT follows from the hierarchical structure shared by all configurations. The second technique goes back to Kari~\cite{Kari1996}: a very simple aperiodic dynamical system is encoded into an SFT, so that configurations correspond to orbits of the dynamical system. Aperiodicity of the SFT comes from both the aperiodicity of the dynamical system and the clever encoding. These two techniques have successfully been generalized to amenable Baumslag-Solitar groups $BS(1,n)$~\cite{AubrunKari2013,aubrun2020tilings}.

\medskip

In this paper we adopt a different strategy and present an innovative construction to lift strongly aperiodic SFTs from a group to a larger group. In particular, we center our attention on the class of Generalized Baumslag-Solitar (GBS) groups, showing that all non-trivial cases admit strongly aperiodic SFTs.

\begin{mycor}{6.11}
All non-$\Z$ Generalized Baumslag-Solitar groups admit a strongly aperiodic SFT.
\end{mycor}

This is done by separately analysing the different quasi-isometry classes of GBS groups, as established by Whyte \cite{whyte2004large}.

\medskip

The paper is organized as follows: Section~\ref{section:symbolic-dynamics} gives the basics of symbolic dynamics on groups. Section~\ref{section:graphs_of_groups} presents graphs of groups and some examples of interest: Torus knot groups and Baumslag-Solitar groups, and Section~\ref{subsection:weak_aper} is devoted to preliminary results that prove the existence of weakly aperiodic SFTs on some classes of graphs of groups. The main result of Section~\ref{section:strong_aper} shows that the property of having a minimal strongly aperiodic SFT is transferred from $H$ to $G$ if $H$ is a normal subgroup of finite index of $G$. Section~\ref{section:Z2} is devoted to the construction of a minimal, strongly aperiodic and horizontally expansive SFT, that will be used later. This construction is a small adaptation of an existing construction presented in~\cite{labbe2022nonexpansive}.

\medskip

The last two sections contain the main results of the paper. In Section~\ref{section:path_foldingFnxZ} we present the key idea of this article: the path-folding technique in the case of $\F_n\times \Z$. The key idea is to fold a $\Z^2$ SFT along a flow on $\F_n$ to obtain an SFT on $\F_n\times \Z$ that shares some dynamical properties with the original SFT. In particular, strong aperiodicity and minimality are preserved.

\begin{mythm}{5.5}
There exists a minimal strongly aperiodic SFT on $\F_n\times\Z$.
\end{mythm}

Then in Section~\ref{section:BS23} we explain how to adapt the path-folding method to the Baumslag-Solitar group $BS(2,3)$ in order to construct a strongly aperiodic SFT in this group. Instead of lifting an aperiodic subshift from $\Z^2$, we codify orbits of a simple dynamical system that ultimately grants the aperiodicity. Consequently, we are able to establish that all non-solvable Baumslag-Solitar groups admit a strongly aperiodic SFT.

\begin{mythm}{6.10}
Non-residually finite Baumslag-Solitar groups $BS(m,n)$ with $m, n > 1$ and $m\neq n$ admit strongly aperiodic SFTs.
\end{mythm}

\section{Subshifts and aperiodicity}\label{section:symbolic-dynamics}

We begin by briefly defining notions from symbolic dynamics needed in this article. For a more comprehensive introduction we refer to~\cite{LindMarcus,CAonGroups,AubrunBarbieriJeandel2018}.

\medskip

Let $A$ be a finite alphabet and $G$ a finitely generated group. We define the \emph{full-shift} on $A$ as the set of configurations $A^G = \{x:G\to A\}$. The group $G$ acts naturally  on this set through the left group action, $\sigma:G\times A^{G}\rightarrow A^{G}$, defined as
    $$\sigma^{g}(x)_{h} = x_{g^{-1}h}.$$
    
By taking the discrete topology on the finite set $A$, the full-shift $A^{G}$ is endowed with the pro-discrete topology, and is in fact a compact set.\\

Let $P$ be a finite subset of $G$. We call $p\in A^{P}$ a \emph{pattern} of support $P$. We say a pattern $p$ appears in a configuration $x\in A^G$ if there exists $g\in G$ such that $p_h = x_{gh}$ for all $h\in P$.\\

Given a set of patterns $\mathcal{F}$, we define the \emph{subshift} $X_{\mathcal{F}}$ as the set of configurations where no pattern in $\mathcal{F}$ appears. That is,
$$X_{\mathcal{F}} = \{x\in A^G \mid p \text{ does not appear in } x,  \ \forall p\in\mathcal{F}\}.$$

It is a well known result that subshifts can be characterized as the closed $G$-invariant subsets of $A^G$. When $\mathcal{F}$ is finite, we say $X_{\mathcal{F}}$ is a \emph{subshift of finite type}, which we will refer to by the acronym SFT.\\

If $x\in A^G$ is a configuration, its \emph{orbit} is the set of configurations 
\[
orb(x):=\left\{\sigma^{g}(x) \mid g\in G \right\}
\]
and its \emph{stabilizer} is the subgroup
\[
stab(x):=\left\{g\in G \mid \sigma^{g}(x) = x \right\}.
\]

A subshift is said to be \emph{minimal} if it does not contain a nontrivial subshift, or equivalently if all orbits are dense.

A subshift is \emph{weakly aperiodic} if all its configurations have infinite orbit. It is \emph{strongly aperiodic} if all its configurations have trivial stabilizer. For a recent survey on the existence of strongly aperiodic SFTs, we refer to~\cite{Rieck2022}.

\medskip

A decision problem that seems strongly related to the existence of aperiodic SFTs is the domino problem, or emptiness problem for SFTs: A group $G$ is said to have decidable domino problem if there is an algorithm that starting from a finite set of forbidden patterns, decides whether the SFT it generates is empty or not. The problem is said to be undecidable if no such algorithm exists. The commonly accepted conjecture states that virtually free groups are precisely those with decidable domino problem. We refer to~\cite{AubrunBarbieriMoutot2019,BartholdiSalo2020} for recent advances on this problem. In this article, we will only make use of the fact that the domino problem is undecidable on $\Z^2$ and that if $H$ is a subgroup of $G$, then if $H$ has undecidable domino problem then so has $G$~\cite{AubrunBarbieriJeandel2018}.

\section{Graphs of groups and GBS}
\label{section:graphs_of_groups}

A common strategy in the study of group theoretical properties is to decompose groups into simpler components and looking at those properties on these simpler groups. HNN-extensions and amalgamated free products are examples of these decompositions. It is along these lines that, to establish the existence of strongly aperiodic SFTs on groups, we make use of the graph of groups decomposition. The Dunwoody-Stallings theorem gives a powerful tool in this regard:

\begin{theorem}[(Dunwoody-Stallings \cite{Dunwoody1985})]
Let $G$ be a finitely presented group. Then, $G$ is the fundamental group of a graph of groups where all edge groups are finite, and vertex groups are either 0 or 1-ended.
\end{theorem}

This approach seems relevant regarding the two problems of characterizing groups which admit strongly aperiodic SFTs or which have decidable domino problem. For instance, examples of this proof technique for characterization of virtually free groups can be seen in \cite{gentimis2008asymptotic, khukhro2020characterisation}.

\subsection{Definition}

For the purposes of this section, we define a graph $\Gamma$ as a tuple $(V_{\Gamma}, E_{\Gamma})$, where $V_{\Gamma}$ is the set of vertices and $E_{\Gamma}\subseteq V_{\Gamma}^2$ is a set of edges, such that the graph is locally finite. We also associate the graph with two functions $\init,\ter:E_{\Gamma}\to V_{\Gamma}$ that give the initial and terminal vertex of en edge, respectively. Given an edge $e\in E_{\Gamma}$, we denote by $\bar{e}$ the edge pointing in the opposite direction to $e$, i.e. $\ter(\bar{e}) = \init(e)$ and $\init(\bar{e})=\ter(e)$.

\begin{defn}
A \emph{graph of groups} $(\Gamma, \G)$ is a connected graph $\Gamma$, along with a collection of groups and monomorphisms $\G$ that includes:
\begin{itemize}
    \item a vertex group $G_v$ for each $v\in V_{\Gamma}$,
    \item an edge group $G_e$ for each $e\in E_{\Gamma}$, where $G_e = G_{\bar{e}}$,
    \item a set of monomorphisms $\{\alpha_{e}:G_e\to G_{\ter(e)} \ | \ e\in E_{\Gamma} \}$.
\end{itemize}
\end{defn}

The main interest of this object is its fundamental group. As its name suggests, this group is obtained through a precise definition of paths on the graph of groups. Luckily there is an explicit expression for the fundamental group, which allows us to skip the formal definition. A complete treatment of the concept can be found in \cite{Lyman_2020, Serre1980}. 

\begin{theorem}
\label{thm:gog}
Let $T\subseteq\Gamma$ be a spanning tree. The group $\pi_1(\Gamma, \G, T)$ is isomorphic to a quotient of the free product of the vertex groups, with the free group on the set $E_{\Gamma}$ of oriented edges. That is,
$$\left(\Conv_{v\in V_{\Gamma}}G_v \ast F(E_{\Gamma})\right)/R,$$
where $R$ is the normal closure of the subgroup generated by the following relations
\begin{itemize}
    \item $\alpha_{\bar{e}}(h)e = e\alpha_{e}(h)$, where $e$ is an oriented edge of $\Gamma$, $h\in G_{e}$,
    
    \item $\bar{e} = e^{-1}$, where $e$ is an oriented edge of $E_{\Gamma}$,
    
    \item $e = 1$ if $e$ is an oriented edge of $T_0$.
\end{itemize}
\end{theorem}

We will omit $\G$ as the context allows it. Furthermore, the fundamental group does not depend on the spanning tree, up to isomorphism.

\begin{prop}[(Proposition 20, \cite{Serre1980})]
The fundamental group of a graph of groups does not depend on the spanning tree.
\end{prop}

Let us look at how traditional operations of geometric group theory are viewed as fundamental groups of graph of groups. 
The amalgamated free product $G\ast_{K}H$ is viewed as the fundamental group $\pi_1(\Gamma_1)$:

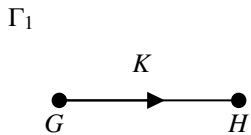
\begin{center}
    \tikzset{every picture/.style={line width=0.75pt}} 

\begin{tikzpicture}[x=0.75pt,y=0.75pt,yscale=-1,xscale=1]

\draw    (230,120) -- (320,120) ;
\draw [shift={(320,120)}, rotate = 0] [color={rgb, 255:red, 0; green, 0; blue, 0 }  ][fill={rgb, 255:red, 0; green, 0; blue, 0 }  ][line width=0.75]      (0, 0) circle [x radius= 3.35, y radius= 3.35]   ;
\draw [shift={(230,120)}, rotate = 0] [color={rgb, 255:red, 0; green, 0; blue, 0 }  ][fill={rgb, 255:red, 0; green, 0; blue, 0 }  ][line width=0.75]      (0, 0) circle [x radius= 3.35, y radius= 3.35]   ;
\draw    (230,120) -- (280,120) ;
\draw [shift={(283,120)}, rotate = 180] [fill={rgb, 255:red, 0; green, 0; blue, 0 }  ][line width=0.08]  [draw opacity=0] (8.93,-4.29) -- (0,0) -- (8.93,4.29) -- cycle    ;

\draw (265,95.4) node [anchor=north west][inner sep=0.75pt]    {$K$};
\draw (221,125.4) node [anchor=north west][inner sep=0.75pt]    {$G$};
\draw (313,125.4) node [anchor=north west][inner sep=0.75pt]    {$H$};
\draw (203,72.4) node [anchor=north west][inner sep=0.75pt]    {$\Gamma _{1}{}$};

\end{tikzpicture}
\end{center}

Similarly, an HNN-extension $G\ast_{\phi}$ can be seen as the fundamental group $\pi_1(\Gamma_2)$:
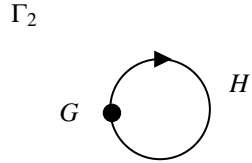
\begin{center}
    \tikzset{every picture/.style={line width=0.75pt}} 

\begin{tikzpicture}[x=0.75pt,y=0.75pt,yscale=-1,xscale=1]

\draw   (254,126) .. controls (254,112.19) and (265.19,101) .. (279,101) .. controls (292.81,101) and (304,112.19) .. (304,126) .. controls (304,139.81) and (292.81,151) .. (279,151) .. controls (265.19,151) and (254,139.81) .. (254,126) -- cycle ;
\draw  [fill={rgb, 255:red, 0; green, 0; blue, 0 }  ,fill opacity=1 ] (250.75,127.88) .. controls (250.75,125.6) and (252.6,123.75) .. (254.88,123.75) .. controls (257.15,123.75) and (259,125.6) .. (259,127.88) .. controls (259,130.15) and (257.15,132) .. (254.88,132) .. controls (252.6,132) and (250.75,130.15) .. (250.75,127.88) -- cycle ;
\draw    (285,101) ;
\draw [shift={(285,101)}, rotate = 180] [fill={rgb, 255:red, 0; green, 0; blue, 0 }  ][line width=0.08]  [draw opacity=0] (8.93,-4.29) -- (0,0) -- (8.93,4.29) -- cycle    ;

\draw (227,121.4) node [anchor=north west][inner sep=0.75pt]    {$G$};
\draw (203,72.4) node [anchor=north west][inner sep=0.75pt]    {$\Gamma _{2}{}{}$};
\draw (312,107.4) node [anchor=north west][inner sep=0.75pt]    {$H$};

\end{tikzpicture}
\end{center}
with $H$ a subgroup of $G$, $\alpha_e = \textnormal{id}$ and $\alpha_{\bar{e}} = \phi:H\to\phi(H)$ an isomorphism. In this sense, the concept of graph of groups can be seen as the natural generalization of these concepts.\\

This article is primarily concerned with a specific class of graph of groups.

\begin{defn}
\label{def:gbs}
A group $G$ is said to be a \emph{Generalized Baumslag-Solitar group} (GBS group) if it is the fundamental group of a finite graph of groups where all the vertex and edge groups are $\Z$.
\end{defn}

As their name suggest GBS groups were introduced as a generalization of Baumslag Solitar groups. This class also contains all torus knot groups as well as $\Z$, $\Z^2$ and the fundamental group of the Klein bottle. For an extensive introduction to this class see \cite{forester2003uniqueness, forester2006splittings, levitt2007automorphism}.

\subsection{Torus knot groups}

The $(n,m)$-torus knot group is given by the presentation 
$$\Lambda(n,m) = \langle a,b \ | \ a^nb^{-m}\rangle.$$

As metioned above they are a particular case of the Generalized Baumslag-Solitar groups, given by the amalgamated free product $\Z\ast_{\Z}\Z$, along with the inclusions $1\mapsto n$ and $1 \mapsto m$. Their name comes from the fact that they are the knot groups of torus knots \cite{rolfsen2003knots}. \\

Remark that $\Lambda(n,m)\simeq\Lambda(m,n)$. Also, if $n$ or $m$ is equal to one, then $\Lambda(n,m) \simeq\Z$, and the group is therefore abelian. In fact, these are the only cases where the group is amenable as a consequence of the following Proposition.

\begin{prop}

$\Lambda(n,m)$ has a finite index normal subgroup isomorphic to $\mathbb{F}_{(n-1)(m-1)}\times\Z$.
\end{prop}
This fact is deduced from the short exact sequence
$$1 \to \Z \to \Lambda(n,m) \to \Z/n\Z\times\Z/m\Z\to 1.$$

As we will later see, these groups are part of a larger subclass of GBS groups having this property. 

\begin{prop}
    For $n,m\geq 2$, $\Lambda(n,m)$ has undecidable domino problem.
\end{prop}

This is a consequence of the fact that these groups contain isomorphic copies of $\Z^2$, namely $\langle a^n, ba\rangle$.\\

\subsection{Baumslag-Solitar groups}

Baumslag-Solitar groups were introduced as we now know them in \cite{baumslag1962some}, to provide examples of non-Hopfian groups, although cases of them were defined some years prior by Higman in \cite{higman1951finitely}. They are simple cases of HNN extensions, and have provided discriminating examples in both combinatorial and geometric group theories. Concerning the study of SFTs on groups, they are also of particular interest since as one-relator groups, they have decidable word problem. Since they are also 1-ended, they fall under the scope of Conjecture~\ref{conjecture:aperiodic_SFT}.

\medskip

Baumslag-Solitar groups are defined by the presentation:

$$BS(m,n) = \langle a, t \mid t^{-1}a^m t = a^n \rangle.$$

The first things to note are that $BS(1,1) = \Z^2$ and $BS(m,n) \simeq BS(-m,-n)$. Baumslag-Solitar groups may behave radically differently: the groups $BS(1,n)$ are solvable and amenable while the $BS(m,n)$ groups with $m,n>1$ contain free subgroups 
and are consequently non-amenable. This dichotomy is also present in the classification of Baumslag-Solitar groups up to quasi-isometry. On the one hand groups $BS(1,n)$ and $BS(1,n')$ are quasi-isometric if and only if $n$ and $n'$ have a common power~\cite{FarbMosher1998} --and in this case, the two groups are even commensurable. On the other hand groups $BS(m,n)$ and $BS(m',n')$ are quasi-isometric as soon as $2\leq m < n$ and $2\leq m'<n'$~\cite{whyte2004large}.

\subsection{Weakly aperiodic SFTs on groups generated from graph structures}\label{subsection:weak_aper}

\begin{prop}\label{prop:weak_aper_graphs_of_groups}
Let $(\Gamma, \G)$ be a graph of groups. If at least one vertex group admits a weakly aperiodic SFT, then $\pi_1(\G)$ admits a weakly aperiodic SFT.
\end{prop}

\begin{proof}
 Theorem \ref{thm:gog} tells us that for every $v\in V_{\Gamma}$, there is a natural injective homomorphism $G_v\hookrightarrow \pi_1(\G)$. Because there is at least one $G_v$ that admits a weakly aperiodic SFT, and weakly aperiodic SFTs can be lifted from subgroups, we conclude that $\pi_1(\G)$ admits a weakly aperiodic SFT.
\end{proof}

One case that does not fall within the hypothesis of Proposition~\ref{prop:weak_aper_graphs_of_groups} is when all vertices of the graph have $\Z$ as their vertex group, which is known not to admit any weakly aperiodic SFT. But a careful study shows that is this case, weakly aperiodic SFT can nevertheless be constructed unless the group is $\Z$ itself.

\begin{prop}
If $\G$ is a graph of $\Z$'s such that $\pi_1(\G)$ is not $\Z$, then $\pi_1(\G)$ has a weakly aperiodic SFT.
\end{prop}

\begin{proof}
Let $G$ be a GBS group with its corresponding graph of groups $\Gamma$. Because $G$ is not $\Z$, at least one edge, $e\in E_{\Gamma}$, satisfies $\alpha_{e}\not\equiv \pm1$. If this edge is a loop, from the previous remarks we know that $G$ contains a non-$\Z$ Baumslag-Solitar group. These groups are known to admit weakly aperiodic SFT \cite{AubrunKari2013}, so we are done. Similarly, if the edge is in the spanning tree $T\subseteq\Gamma$ such that $G = \pi_{1}(\Gamma, T)$, then $G$ contains a knot group $\Lambda(n,m)$, which admits a weakly aperiodic SFT by virtue of containing $\Z^2$ as a subgroup.

The last case is when all edges in the spanning tree satisfy $\alpha_{e'}\equiv\pm1$, and there are no loops. Let $e$ be an edge such that $\alpha_{e},\alpha_{\bar{e}}\neq \pm 1$, and $v,u$ its end points. Because $T$ is spanning, we know that $v,u\in V_{T}$, and therefore if $G_v = \langle a \rangle$ and $G_u = \langle b \rangle$ we have that in $G$, $a = b^{\pm 1}$. Then, the relation given by the edge $e$ is,
$$a^{\alpha_{e}(1)}e = eb^{\alpha_{\bar{e}}(1)} \ \iff \ a^{\alpha_{e}(1)}e = ea^{\pm\alpha_{\bar{e}}(1)}.$$

This means $G$ contains the non-$\Z$ Baumslag Solitar group $BS(\alpha_{e}(1), \pm\alpha_{\bar{e}}(1))$, which as mentioned before, admits a weakly aperiodic SFT. 
\end{proof}

\begin{rem}
Notice that this proof also shows that all non-$\Z$ GBS groups have undecidable domino problem.
\end{rem}

The same proof scheme can be utilized for the class of Artin groups, which are another example of groups generated from an underlying graph structure.\\

Let $\Gamma = (V, E, \lambda)$ be a labeled graph with labels $\lambda: E\to \{2,3, ...\}$. We define the \emph{Artin group} of $\Gamma$ through the presentation:

$$A(\Gamma) := \langle V \ | \ \underbrace{abab...}_{\lambda(e)} = \underbrace{baba...}_{\lambda(e)}, \ \forall e = (a,b)\in E\rangle.$$

Let us call $\Gamma_n$ be the graph of 2 vertices $a$ and $b$ and the edge connecting them labeled by $n$. Notice that $A(\Gamma_2) \simeq \Z^2$.

\begin{prop}
All non-$\Z$ Artin groups admit a weakly aperiodic SFT.
\end{prop}

\begin{proof}
Let $A(\Gamma)$ be an Artin group defined from $\Gamma = (V, E, \lambda)$. If $E$ is empty then $A(\Gamma)$ is the free group of rank $|V|\geq 2$, which is known to admit weakly aperiodic SFTs \cite{piantadosi2008symbolic}. 

Let $e=(a,b)$ be an edge in $E$. Notice that $A(\Gamma_{n}) \simeq\langle a, b\rangle\leq A(\Gamma)$. Because weakly aperiodic SFTs are inherited from subgroups, it suffices to show that $A(\Gamma_n)$ admits a weakly aperiodic SFT for every $n\in\N$. We identify two cases:

\begin{itemize}
    \item Case 1: $n = 2k$, $k\geq 1$.\\
    
    We have that $A(\Gamma_{2k})$ is the one-relator group:
    $$A(\Gamma_{2k}) = \langle a,b \ | \ (ab)^{k} = (ba)^k\rangle = \langle a,b \ | \ (ab)^{k} = b(ab)^{k}b^{-1}\rangle .$$
    
    We apply Tietze transformations to the presentation,
    \begin{align*}
        A(\Gamma_{2k}) &\simeq \langle a,b, c \ | \ (ab)^{k} = b(ab)^{k}b^{-1}, \ c = ab\rangle\\
        &\simeq \langle b,c \mid b^{-1}c^{k}b = c^{k}\rangle \\
        &= BS(k,k)
    \end{align*}
    Therefore, $A(\Gamma_{2k})$ admits a weakly aperiodic SFT.
    
    \item Case 2: $n = 2k + 1$, $k\geq 1$.\\
    
    Once again, $A(\Gamma_{2k+1})$ is the one-relator group:
    $$A(\Gamma_{2k+1}) = \langle a,b \ | \ (ab)^{k}a = (ba)^kb\rangle = \langle a,b \ | \ (ab)^{k}a = b(ab)^{k}\rangle .$$
    By doing an analogous procedure, we arrive at
    $$A(\Gamma_{2k+1})\simeq \Lambda(2, 2k+1).$$
    
    Therefore, $A(\Gamma_{2k+1})$ also admits a weakly aperiodic SFT.
\end{itemize}
\end{proof}

\begin{rem}
Once again, this proof also shows that all non-free Artin groups have undecidable domino problem.
\end{rem}
\section{Strong aperiodicity}\label{section:strong_aper}

\subsection{State of the art}

The existence of an aperiodic SFT is a geometric property of groups, at least for finitely presented ones, as stated below.

\begin{theorem}[(Cohen \cite{cohen2017large})]\label{theorem:Cohen_QI}
Let $G$ and $H$ be two quasi-isometric finitely presented groups. Then $G$ admits a strongly aperiodic SFT if and only if $H$ does.
\end{theorem}

The hypothesis of finite presentation is essential here. For example, from \cite{barbieri2017geometric} we know that the Grigorchuk group, which is finitely generated but not finitely presented, admits a strongly aperiodic SFT. Nevertheless, the Grigorchuk group has uncountably many groups which are quasi-isometric to it. This means that at least one of them has undecidable word problem, and therefore does not admit a strongly aperiodic SFT by \cite{Jeandel2015aperiodic}. This result can also be seen as a consequence of the fact that being recursively presented is not a quasi-isometry invariant.

However, we do have results when two finitely generated groups are commensurable. Two groups are said to be commensurable if they have finite index subgroups which are isomorphic.

\begin{theorem}[(Carroll, Penland \cite{carroll2015periodic})]
Let $G$ and $H$ be two finitely generated groups which are commensurable. Then $G$ admits a strongly aperiodic SFT if and only if $H$ does.
\end{theorem}

Commensurability implies quasi-isometry, but there exists many examples where the converse does not hold. For instance, the groups $BS(m,n)$ and $BS(p,q)$ are quasi-isometric whenever $1<n<m$ and $1<p<q$, but are not commensurable if $(m,n)$ are co-prime and $(p,q)$ are co-prime \cite{whyte2004large, casals2019commensurability}.
  
\subsection{Lifting Minimal Aperiodic Subshifts}\label{subsection:lifting}

The goal of this section is to lift a minimal strongly aperiodic SFT from a normal subgroup of finite index to the whole group. In order to do this we make use of the locked shift, as introduced by Carroll and Penland \cite{carroll2015periodic}.\\

Let $A$ be a finite alphabet, $G$ a finitely generated group with $N$ a finitely generated normal subgroup. We define the subshift

$$\text{Fix}_A(N) = \{x\in A^{G}\mid \sigma^n(x) = x, \ \forall n\in N\}.$$

\begin{rem}
For any subgroup $N$, $\text{Fix}_A(N)$ is always a closed set of $A^G$, but it is only shift invariant when $N$ is normal. In the latter case, the subshift is conjugated to $A^{G/N}$.
\end{rem}

\begin{defn}
For a finite index normal subgroup $N$ we define the $N$-locked subshift $L$ as $\textrm{Fix}_{R}(N)\cap \Sigma$, where $R$ is a set of right coset representatives with $1\in R$, and $\Sigma$ is the subshift defined by the the finite set of forbidden patterns
$$\{p:\{1,r\}\to R \mid r\in R\setminus\{1\}, \ p(1) = p(r) \}.$$
\end{defn}

\begin{lemma}
\label{locked}
The $N$-locked subshift $L$ is a non-empty SFT. In addition, $\sigma^{g}(x) = x$ for some $x\in L$ if and only if $g\in N$.
\end{lemma}

\begin{proof}
If we take $S = \{s_1, ..., s_m\}$ a set of symmetric generators for $N$, we see that $\text{Fix}_R(N)$ is an SFT by the set of forbidden rules given by 
$$\{p:\{1,s_i\}\to R \mid s_i\in S, \ p(1)\neq p(s_i)\}.$$

Therefore, $L$ is also an SFT. To see it is non-empty we define $y\in R^{G}$ by $y_{nr} = r$. If we take $n'\in N$,
$$\sigma^{n'}(y)_{nr} = y_{n'^{-1}n r} = r = y_{nr}.$$

Thus, $y\in\text{Fix}_R(N)$. Next, we take $r'\in R$ and see that

$$y_{nr} = \sigma^{n^{-1}}(y)_r =  y_r \neq y_{rr'} = \sigma^{n^{-1}}(y)_{rr'} = y_{nrr'}.$$

This way, $y\in\Sigma$, and therefore $y\in L$. Finally, if we take $x\in L$ and $g = nr\in G$ such that $\sigma^{g}(x)=x$, then
$$x = \sigma^{nr}(x) = \sigma^{r(r^{-1}nr)}(x) = \sigma^{r}(x).$$

With this, $x_1 = x_{r^{-1}r} = \sigma^r(x)_r = x_r$. Because $x\in\Sigma$, $r=1$ and thus $g = n\in N$.

\end{proof}

We now have all the ingredients to prove the result.

\begin{prop}\label{prop:normal_finite_index_lift}
Let $G$ be a finitely generated group and $H$ a finite index normal subgroup. If $H$ admits a minimal strongly aperiodic SFT, then $G$ also does.
\end{prop}

\begin{proof}
Let $X\subseteq A^{H}$ be a minimal strongly aperiodic SFT over $H$. Given a set $R$ of right coset representatives with $1\in T$, we define the $G$-subshifts

$$\hat{X} = \{y\in A^G\mid \exists x\in X, \ \forall(h,r)\in H\times R, \ y_{hr} = x_h \},$$

and $Y = \hat{X}\times L$, where $L$ is the $H$-locked shift. Let us see that $Y$ is the subshift we are looking for
\begin{itemize}
    \item \textbf{Aperiodicity: } Suppose there is a $y\in Y$ and $g\in G$ such that $\sigma^{g}(y) = y$. Due to Lemma \ref{locked} we know that $g\in H$. Then, for some $r\in R$ and $x\in X$
    
    $$x_{h} = y_{hr} = y_{g^{-1}hr} = x_{g^{-1}h},$$
    
    that is, $x = \sigma^{g}(x)$ which contradicts the aperiodicity of $X$.
    
    \item \textbf{Minimality: } Let us take $y, y'\in Y$ along with $x,x'\in X$ their corresponding $X$ configurations. By the minimality of $X$, there exists a sequence $\{h_n\}_{n\in\N}\subseteq H$ such that $\sigma^{h_n}(x)\to x'$. Then, for $(h,r)\in H\times R$
    
    $$\sigma^{h_n}(y)_{hr} = y_{h_n^{-1}hr} = x_{h_n^{-1}h} = \sigma^{h_n}(x)_h \to x'_{h} = y'_{hr}.$$
    
    Thus $\sigma^{h_n}(y)\to y'$.
\end{itemize}

\end{proof}

\subsection{Strong aperiodicity for GBS}

The quasi-isometric structure of GBS groups is well understood: there are classified according to the following result.

\begin{theorem}[(Whyte \cite{whyte2004large})]\label{theorem:Whyte}
If $\G$ is a graph of $\Z$'s, then for $G = \pi_1(\G)$ exactly one of the following is true:
\begin{enumerate}
    \item $G$ contains a finite index subgroup isomorphic to $\F_n\times\Z$,
    \item $G = BS(1,n)$ for some $n>1$,
    \item $G$ is quasi-isometric to $BS(2,3)$.
\end{enumerate}
\end{theorem}

GBS groups that fall into the first category are called \emph{unimodular}. In ~\cite{Jeandel2015aperiodic} Jeandel show that $\F_2\times\Z$ admits a strongly aperiodic SFT through the use of Kari-like tiles. Nevertheless this construction is not minimal and is only valid on $\F_2\times\Z$ and not all $\F_n\times\Z$ for $n\geq3$. As explained in Section~\ref{subsection:lifting}, for a group $G$, containing a finite index subgroup that admit a minimal strongly aperiodic SFT is not enough to get a minimal strongly aperiodic SFT on $G$. Fortunately the statement of Theorem~\ref{theorem:Whyte} can be sharpened, as stated in the remark below.

\begin{rem}
In the case of unimodular GBS it can even be proven that $G$ contains $\F_n\times\Z$ as a normal subgroup of finite index~\cite[Lemma 4]{DelgadoRobinsonRimm2017}. The additional normality hypothesis will be needed in Proposition~\ref{prop:normal_finite_index_lift}.
\end{rem}

In Section~\ref{section:path_foldingFnxZ} we provide an example of a minimal strongly aperiodic SFT on $\F_n\times\Z$ for all $n\geq2$. The case of groups quasi-isometric to $BS(2,3)$ is also treated in this paper: in Section~\ref{section:BS23} we explain how to construct a strongly aperiodic SFT on $BS(2,3)$. Groups $G = BS(1,n)$ for some $n>1$ are already known to possess a minimal strongly aperiodic SFT~\cite{aubrun2020tilings}. In total we are able to construct strongly aperiodic SFTs for all GBS.

\section{A minimal, strongly aperiodic and horizontally expansive SFT on $\Z^2$}\label{section:Z2}

In this section we present a construction of a strongly aperiodic SFT on $\Z^2$ with additional properties, that will be useful in Section~\ref{subsection:structure_path-folding}. We begin be presenting the notion of expansive subspaces or directions as introduced in \cite{boyle1997expansive}. Let $F$ be a subspace of $\R^2$ and $v\in\R^{2}$, we define
    $$\dist(v,F) = \inf\{\|v-w\|: w\in F\},$$
where $\|\cdot\|$ denotes the Euclidean norm on $\R^{2}$. For $t>0$ we define the thickening of $F$ by $t$ as $F^{t} = \{v\in\Z^{2}: \dist(v,F)\leq t\}$. We say a subspace $F$ is expansive for a subshift $X$ if there exists $t > 0$ such that for any two configurations $x,y\in X$, $x|_{F^t} = y|_{F^t}$ implies $x = y$. Conversely, $F$ is said to be non-expansive if for all $t > 0$ there exist distinct $x,y\in X$ such that $x|_{F^t} = y|_{F^t}$.

As we are working with two dimensions, non-trivial subspaces can be represented by directions. Thus we speak of expansive and non-expansive directions.\\

For our purposes, a subshift $X\subset A^{\Z^2}$ is \emph{horizontally expansive} --resp. \emph{vertically expansive}-- if for every pair of configurations $x,y$ in $X$ if $x_{\Z\times \{0\}}=y_{\Z\times \{0\}}$ --resp. $x_{\{0\}\times \Z}=y_{\{0\}\times \Z}$--, then $x=y$. Said otherwise, one single row entirely determines the global configuration in the subshift. To construct our sought after SFT, we can for instance make use of the following construction.

\begin{theorem}[(Labb\'e, Mann, McLoud-Mann \cite{labbe2022nonexpansive}, Labb\'e \cite{labbe2021rauzy, labbe2021markov})]\label{theorem:SFT_Labbe}
There exists an aperiodic, minimal SFT $X_0$ such that its non-expansive directions are given by the lines of slope $\{0, \varphi+3, 2-3\varphi, \frac52 - \varphi\}$, where $\varphi = \frac{1+\sqrt{5}}{2}$ is the golden mean.
\end{theorem}

In particular, this result tells us that the vertical line is an expansive direction for $X_0$. It suffices to convert this expansive direction into horizontal expansivity to get the desired SFT, as we will show in what follows. Notice that we can take $X_0$ to be a Wang tile SFT by taking a higher block shift. This process preserves expansive directions as stated in the next result.

\begin{lemma}[(\cite{labbe2022nonexpansive})]
Let $X$ and $Y$ be two conjugated $\Z^2$-subshifts and $v\in\R^2$. Then, $v$ is a non-expansive direction for $X$ if and only if it is non-expansive for $Y$.
\end{lemma}

Moreover, up to another conjugacy --a higher block again in this case-- we can also impose that the thickening $t$ of an expansive direction is zero. Then we get the following.

\begin{lemma}\label{lemma:conjugate_expansive_thickening_zero}
Let $X$ be a $\Z^2$-subshift and $v\in\R^2$ an expansive direction for $X$. Then there exists $Y$ a $\Z^2$-subshift conjugate to $X$ such that $Y$ is expansive in direction $v$ with thickening $t=0$.
\end{lemma}


Now that the SFT $X_0$ from~\cite{labbe2022nonexpansive} has been converted thanks to Lemma~\ref{lemma:conjugate_expansive_thickening_zero} into a conjugated vertically expansive Wang tile SFT $Y_0$, we can rotate its Wang tiles and thus its configurations by $\frac{\pi}{2}$ (see Figure~\ref{fig:rotate_Wang_tile}). This rotated tileset defines an SFT, called the \emph{rotation by $\frac{\pi}{2}$} of $Y_0$. 

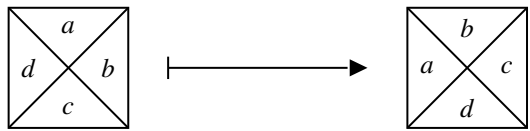
\begin{figure}[H]
\begin{center}
    \tikzset{every picture/.style={line width=0.75pt}} 
\begin{tikzpicture}[x=0.75pt,y=0.75pt,yscale=-1,xscale=1]
\draw   (200,100) -- (260,100) -- (260,160) -- (200,160) -- cycle ;
\draw    (200,100) -- (260,160) ;
\draw    (260,100) -- (200,160) ;
\draw    (280,130) -- (370,130) ;
\draw [shift={(380,130)}, rotate = 180] [fill={rgb, 255:red, 0; green, 0; blue, 0 }  ][line width=0.08]  [draw opacity=0] (8.93,-4.29) -- (0,0) -- (8.93,4.29) -- cycle    ;
\draw [shift={(280,130)}, rotate = 180] [color={rgb, 255:red, 0; green, 0; blue, 0 }  ][line width=0.75]    (0,5.59) -- (0,-5.59)   ;
\draw   (400,100) -- (460,100) -- (460,160) -- (400,160) -- cycle ;
\draw    (400,100) -- (460,160) ;
\draw    (460,100) -- (400,160) ;
\draw (230,110) node    {$a$};
\draw (250,130) node    {$b$};
\draw (230,150) node    {$c$};
\draw (210,130) node    {$d$};
\draw (430,110) node    {$b$};
\draw (450,130) node    {$c$};
\draw (430,150) node    {$d$};
\draw (410,130) node    {$a$};

\end{tikzpicture}    
    \caption{A Wang tile and its rotation by $\frac{\pi}{2}$.}
    \label{fig:rotate_Wang_tile}
\end{center}
\end{figure}

\begin{lemma}\label{lemma:rotate_SFT}
Let $X$ be a minimal, strongly aperiodic and vertically expansive Wang tile SFT. Then its rotation by $\frac{\pi}{2}$ is a minimal, strongly aperiodic and horizontally expansive Wang tile SFT.
\end{lemma}

The proof of Lemma~\ref{lemma:rotate_SFT} does not pose any specific difficulties and is thus omitted. Combining this result with Theorem~\ref{theorem:SFT_Labbe} we conclude that there exists a minimal, strongly aperiodic and horizontally expansive Wang tile SFT. This result will be used in Section~\ref{subsection:minimality}.

\begin{prop}\label{prop:minimal_sa_H_expansive_SFT}
There exists a minimal, strongly aperiodic and horizontally expansive Wang tile SFT.
\end{prop}

\section{The path-folding technique on $\F_n\times \Z$}\label{section:path_foldingFnxZ}

In this section we present a technique to convert a subshift on $\Z^2$ into a subshift on $\F_n\times\Z$ that shares some of its properties: the \emph{path-folding} technique. In our case the properties that are proven to be preserved are: being of finite type (SFT), strong aperiodicity and minimality. In this section we use $\pi_1$ as the projection onto the first coordinate, and not as a fundamental group.

As we will see later, this technique has a broader scope. In its most abstract version it consists on the following steps:
\begin{enumerate}
    \item Find a regular tree-like structure in the group. In the case of $BS(2,3)$ we take its Bass-Serre tree, and in the case we of $\F_n\times\Z$ simply take $\F_n$.
    
    \item We define the flow shift on the tree: using an alphabet of arrows of the same size as the degree of the vertices, we define a local rule demanding that, for every vertex, only one of its neighbours has an arrow pointing away from the vertex, and the rest pointing towards. This allows us to make a correspondence between the elements of the flow shift and the boundary of the tree.
    
    \item Finally, we fold configurations from other structures along the directions provided by the flow shift. In the case of $BS(2,3)$ we fold configurations from the hyperbolic plane, and for $\F_n\times\Z$ we fold configurations from $\Z^2.$
\end{enumerate}

\subsection{The Flow SFT}
\label{sec:flow}

Let us begin by introducing the \emph{flow shift} over $\F_n\times\Z$, which we will denote as $Y_{f}$. We define this shift from tiles representing different directions. Let $S = \{s_1, ..., s_{n}\}$ be a set of generators of $\F_n$ and $\Z = \langle t\rangle$. We will define the flow shift over the alphabet $A \coloneqq S\cup S^{-1}$. We can interpret these tiles as pointing in the direction specified by a generator or its inverse. 

\begin{figure}[H]
\begin{center}
    \tikzset{every picture/.style={line width=0.75pt}} 

\begin{tikzpicture}[x=0.75pt,y=0.75pt,yscale=-1,xscale=1]

\draw   (120,150) -- (130,140) -- (140,150) -- (160,150) -- (160,170) -- (150,180) -- (160,190) -- (160,210) -- (140,210) -- (130,200) -- (120,210) -- (100,210) -- (100,190) -- (110,180) -- (100,170) -- (100,150) -- cycle ;
\draw    (130.33,188.33) -- (130.04,166.33) ;
\draw [shift={(130,163.33)}, rotate = 89.24] [fill={rgb, 255:red, 0; green, 0; blue, 0 }  ][line width=0.08]  [draw opacity=0] (8.93,-4.29) -- (0,0) -- (8.93,4.29) -- cycle    ;

\draw   (255,170) -- (265,180) -- (255,190) -- (255,210) -- (235,210) -- (225,200) -- (215,210) -- (195,210) -- (195,190) -- (205,180) -- (195,170) -- (195,150) -- (215,150) -- (225,160) -- (235,150) -- (255,150) -- cycle ;
\draw    (216.67,180.33) -- (238.67,180.04) ;
\draw [shift={(241.67,180)}, rotate = 179.24] [fill={rgb, 255:red, 0; green, 0; blue, 0 }  ][line width=0.08]  [draw opacity=0] (8.93,-4.29) -- (0,0) -- (8.93,4.29) -- cycle    ;

\draw   (330,210) -- (320,220) -- (310,210) -- (290,210) -- (290,190) -- (300,180) -- (290,170) -- (290,150) -- (310,150) -- (320,160) -- (330,150) -- (350,150) -- (350,170) -- (340,180) -- (350,190) -- (350,210) -- cycle ;
\draw    (319.67,171.67) -- (319.96,193.67) ;
\draw [shift={(320,196.67)}, rotate = 269.24] [fill={rgb, 255:red, 0; green, 0; blue, 0 }  ][line width=0.08]  [draw opacity=0] (8.93,-4.29) -- (0,0) -- (8.93,4.29) -- cycle    ;

\draw   (385,190) -- (375,180) -- (385,170) -- (385,150) -- (405,150) -- (415,160) -- (425,150) -- (445,150) -- (445,170) -- (435,180) -- (445,190) -- (445,210) -- (425,210) -- (415,200) -- (405,210) -- (385,210) -- cycle ;
\draw    (423.33,179.67) -- (401.33,179.96) ;
\draw [shift={(398.33,180)}, rotate = 359.24] [fill={rgb, 255:red, 0; green, 0; blue, 0 }  ][line width=0.08]  [draw opacity=0] (8.93,-4.29) -- (0,0) -- (8.93,4.29) -- cycle    ;
\end{tikzpicture} 
    \caption{Flow tiles for $\F_2$.}
\end{center}
\end{figure}
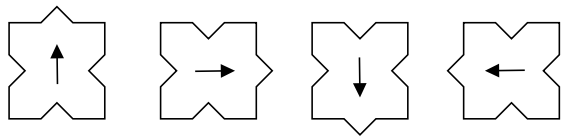

To define $Y_{f}$, we demand that a configuration $y\in A^{\F_n\times\Z}$ satisfies:

$$y_{g} = s \ \implies \begin{cases}y_{gs}\neq s^{-1} \\ 
y_{gs'} = s'^{-1}, \ \forall s'\in A\setminus\{s\}\\
y_{gt} = s\end{cases}.$$

Notice that fixing a tile at the identity completely determines the tiling of the $2n-1$ subtrees of $\F_n$ the tile does not point towards. Then, this leaves $2n-1$ possible tiles for the unspecified neighbour. In addition, the last rule makes sure that each $\Z$-coset contains the same tile.

This allows us to describe each configuration with an infinite word $W$. Given $y\in Y_f$, we recursively define $W(y)\in A^{\N}$ by setting $W_0 \coloneqq y_1$ and setting $W_{n+1} \coloneqq y_{W_0 ... W_n}$.

\begin{figure}[ht]
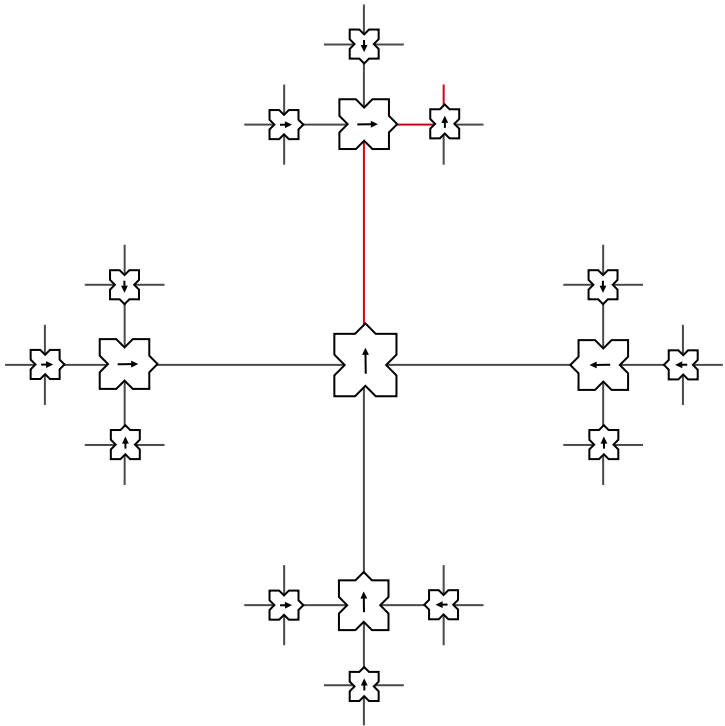

\begin{center}
    \includestandalone[scale=1]{ fig/flowexample}    
    \caption{If we look at the configuration in $\F_2 = \langle a, b\rangle$ as shown, the prefix of its infinite word is given by $bab$.}
\end{center}
\end{figure}

Due to the local rules, this correspondence between configurations and infinite words effectively creates a bijection $W$ between $Y_{f}$ and $\partial_{\infty}\F_{n}$, the boundary of $\F_n$. \\

\begin{prop}\label{prop:flow_period}
If $y\in Y_f$ has period $g\in \F_n$, then $W(y)$ is either the infinite word~$g^\N$ or the infinite word~$(g^{-1})^\N$.
\end{prop}

\begin{proof}
Let $y\in Y_f$ be a configuration with a period $g\in \F_2$, that is, $\sigma^{g}(y)=y$. We get right away that $y_1=y_g$. Let us write $g$ as a reduced word $g_1\dots g_k$ on $\{s_1^{\pm 1}, ..., s_{n}^{\pm 1}\}$. If we assume that $y_1\neq g_1^{\pm 1}$, then following the path from $1$ to $g$, we get that $y_{g_1}=g_1^{-1}$, $y_{g_1g_2}=g_2^{-1}$, \dots as well as that $y_{g}=g_k$, $y_{g_1\dots g_{k-1}}=g_{k-1}$ and so on, by following the path in the opposite direction. So there necessarily exists an index $i$ such that $y_{g_1\dots g_i}=g_{i}^{-1}$ and $y_{g_1\dots g_i}=g_{i}$, which is not possible, hence $y_1= g_1^{\pm 1}$. 
Iterating this process we conclude that either $y_{g_1\dots g_i}=g_{i}^{-1}$ for each $i=1,\dots, k$ or $y_{g_1\dots g_i}=g_{i}$ for each $i=1,\dots, k$. Thus $W(g)$ has either $g$ or $g^{-1}$ as a prefix. By applying the same reasoning to $\sigma^{g}(y)$, $\sigma^{g^2}(y)$, \dots, all of which also admit $g$ as a period, we conclude that either $W(y)=g^\N$ or $W(y)=(g^{-1})^\N$. 
\end{proof}

\subsection{The structure of a path-folding SFT}\label{subsection:structure_path-folding}

Let $X\subseteq A^{\Z^{2}}$ be an horizontally expansive, strongly aperiodic SFT on $\Z^2$; for instance, the SFT detailed in Proposition~\ref{prop:minimal_sa_H_expansive_SFT}. Without loss of generality we assume $X$ to be a nearest neighbor SFT. We want to "fold" each configuration along the path defined by the infinite word of a configuration in $Y_{f}$. Let $Z$ be the subshift of $A^{\F_n\times \Z}\times Y_f$, given be the following set of allowed patterns:

\begin{itemize}
    \item For each valid pattern $H$ of support $\{(0,0), (1,0)\}$ in $X$, we define the pattern $P$ of support $\{1, t\}$ by:
    $$P(1) = (H_{(0,0)}, d), \  P(t) = (H_{(1,0)}, d)$$
    where $d\in T$.
    
    \item For each valid pattern $V$ of support $\{(0,0), (0,1)\}$ in $X$, we define the patterns $Q$ of support $\{1, s\}$ by:
    $$Q(1) = (V_{(0,0)}, s), \  Q(s) = (V_{(1,0)}, s')$$
    where $s'\in T\setminus\{s^{-1}\}$.
\end{itemize}

\begin{prop}
\label{configsJuntas}
The configurations in $Z$ have the following structure:
$$x\otimes y: wt^{i} \to (x_{(i,j)}, \ y_w),$$
where $w\in \F_n$, $x\in X$, $y\in Y_{f}$ defined by the word $W$, with
$$j = 2\max\{|w'| \ : \ w'\sqsubseteq w \wedge w'\sqsubseteq W\} - |w|.$$
\end{prop}

\begin{proof}
Let us have $y\in Y_f$ and $x\in X$. We begin by showing that $x\otimes y\in Z$. We know the second coordinate satisfies the allowed patterns by definition, so we must look at the first.

Let $g = wt^{i}\in \F_n\times\Z$, then
$$x\otimes y(g) = (x_{(i,j)}, \ y_g),$$
as in the definition. We begin by looking at the support $\{1, t\}$. We have that $gt = wt^{i+1}$. Because $w$ doesn't change when adding $t$, $j$ does not change and $y_{wt^{i+1}} = y_{wt^i}$.  Therefore, 
$$\{x\otimes y(g), x\otimes y(gt)\} = \{(x_{(i,j)}, y_w), (x_{(i+1, j)}, y_{w})\}$$

is allowed for all $g\in \F_n\times\Z$. For patterns of support $\{1,s\}$, with $s\in S\cup S^{-1}$, we have $gs = wst^{i}$. Let us denote $u = \argmax\{|w'| \ : \ w'\sqsubseteq w \wedge w'\sqsubseteq W\}$, $j= 2|u| - |w|$ and $\pi_1(x\otimes y)_{gs}= x_{(i,j')}$. If it happens that $y_w = s$, we have two cases:

\begin{itemize}
    \item $u = w$. Then, by applying $s$ we continue on the configurations path, ie, $ws\sqsubseteq W$. Therefore, $j' = j + 1$
    
    \item $u \sqsubset w$. Then, because of the local rules defining $Y_f$, the last letter in $w$ must be $s^{-1}$. Then $|ws| = |w| - 1$ and therefore $j' = j+1$.
\end{itemize}

This means $\{x\otimes y(g), x\otimes y(gs)\} = \{(x_{(i,j)}, y_w), (x_{(i, j+1)}, y_{w})\}$ is allowed. If on the other hand, $y_w \neq s$, we have that 
$$\argmax\{|w'| \ : \ w'\sqsubseteq ws \wedge w'\sqsubseteq W\} = u.$$

Thus, $j' = 2|u| - |ws| = j -1$. Once again, this means $\{x\otimes y(g), x\otimes y(gs)\} = \{(x_{(i,j)}, y_w), (x_{(i, j+1)}, y_{w})\}$ is allowed. We conclude that $x\otimes y\in Z$.\\

Now, let us have $z\in Z$. We can easily obtain $y\in Y_f$ linked to a word $W$ through the recursive method mentioned above. To find $x$, we begin by setting:
$$x_{(i,0)} = \pi_1(z_{t^{i}}), \ \forall i\in\Z.$$

Next, we define the path function $\rho:\Z\to G$ as follows: 

$$\rho_W(j) = \begin{cases} W_0 ... W_j, \ \text{ if } \ j\geq 0 \\ (W_0)^{j} \ \text{ if } \  j < 0 \end{cases}.$$

We continue looking our configuration $y$ by defining the group elements $\{g_{i,j}\}_{i,j\in\Z}$ as $g_{i,j} = \rho_W(j)t^{i}$.

Finally, we set
$$x_{(i,j)} = \pi_1(z_{g_{i,j}}).$$

\textbf{Claim: } $x\in X$.\\

Let us take a look at two cases:
\begin{itemize}
    \item $\exists (i,j)\in\Z^2$: $\{x_{(i,j)}, x_{(i+1,j)}\}$ is forbidden in $X$. This would mean that the pattern $\{z_{g_{i,j}}, z_{g_{i,j}t}\}$ would be forbidden in $Z$, which is a contradiction.

    \item $\exists (i,j)\in\Z^2$: $\{x_{(i,j)}, x_{(i,j+1)}\}$ is forbidden in $X$.
    
    Notice that, $g_{r,n+1} = g_{r,n}W_{n+1}$. This would mean that the pattern $\{z_{g_{i,j}}, z_{g_{i,j}W_{n+1}}\}$ would be forbidden in $Z$, which is a contradiction.
\end{itemize}

\textbf{Claim: } $z = x\otimes y$.

Because of the way $y$ was obtained, it suffices to check the first coordinate. Let us have $g = wt^i$ and $\pi_1(x\otimes y)_g = x_{(i,j)}$ as in the proposition statement. In addition, let 
$$u = \argmax\{|w'| \ : \ w'\sqsubseteq w \wedge w'\sqsubseteq W\},$$

and $N = |u|$. As we have seen, this means that 
$$\pi_1(z_{h}) = x_{(0, N)},$$

because $u = \rho_W(N)$. Now, if we have $w = u w_0 ... w_m$,we can see that  $y_u$ is not in the direction of the flow. Thus, we can deduce from the allowed local rules that the second coordinate of $x$ must decrease by 1 when applying $w_0$. Now, because $X$ is expansive, we know that $x$ is the only configuration with the pattern $x|_{\Z\times\{N\}}$ on $\Z\times \{N\}$. This allows us to say, 

$$\pi_{1}(z_{hw_0}) = x_{\left(0, N-1\right)},$$

By repeating the same argument for $w_{1}$ up to $w_{m}$, we obtain:

$$\pi_{1}(z_{w}) = \pi_{1}(z_{hw_0...w_{m}}) = x_{\left(0, \ N-(m-N)\right)} = x_{(0, \ j)}$$

we can conclude,

$$\pi_1(z_g) = \pi_1(z_{wt^{i}}) = x_{(i,j)}.$$
\end{proof}

\begin{theorem}\label{theorem:strongly_aper_FnZ}
There exists a strongly aperiodic SFT on $\F_n\times\Z$.
\end{theorem}

\begin{proof}
We proceed by contradiction. Let $z\in Z$ be such that there exists $g\in\F_n\times\Z\setminus\{1\}$ satisfying $\sigma^{g}(z) = z$. We decompose $g^{-1}$ as $wt^{i}$, with $w\in\F_n$ and $i\in\Z$. In addition, let us have $x\in X$ and $y\in Y_{f}$ such that $z = x\otimes y$.\\
By Proposition \ref{prop:flow_period}, $W = W(y)$ is a periodic word given by either $w^{\N}$ or $(w^{-1})^{\N}$. Let us call $l = |w|$, and suppose without loss of generality that $W = w^{\N}$.\\

\textbf{Claim: } $\sigma^{(-i,-l)}(x) = x$.\\

Let $(\alpha, \beta)\in\Z^{2}$ and let $h = \rho_{W}(\beta)t^{\alpha}$.  Then, $x_{(\alpha, \beta)} = \pi_1(z_h)$.

If we call $\pi_1(\sigma^{g}(z)_h) = \pi_1(z_{g^{-1}h}) = x_{(\alpha', \beta')}$, it is straightforward to see that $\alpha' = \alpha + i$. For the second coordinate, notice that for $g^{-1}h$ the greatest prefix this element has in common with $W$ is given by $w\rho_{W}(\beta)$, due to the definition of $W$. This means that, $\beta' = \beta + l$, and thus $x\in X$ is periodic in the direction $(-i,-l)$, which is a contradiction.
\end{proof}

\subsection{Minimality}\label{subsection:minimality}

We would like to see if properties from the aperiodic SFT $X$ over $\Z^2$ can be lifted to our new aperiodic subshift $Z$. In particular, we are interested in preserving minimality. A $\Z^2$-SFT of the sought after characteristics is shown to exists in Proposition~\ref{prop:minimal_sa_H_expansive_SFT}.

The idea is as follows. First, we show that the flow shift $Y_f$ is minimal. The idea here is, for configurations defined by words $W'$ and $W$, to shift the fist configuration progressively obtaining configurations whose defining word is $W_0W_1 \ ... \ W_n e_n W'$, where $e_n$ is an error term of length 1. Second, we couple this minimality with that of $X$ to establish the sought after result. 

\begin{figure}[h]
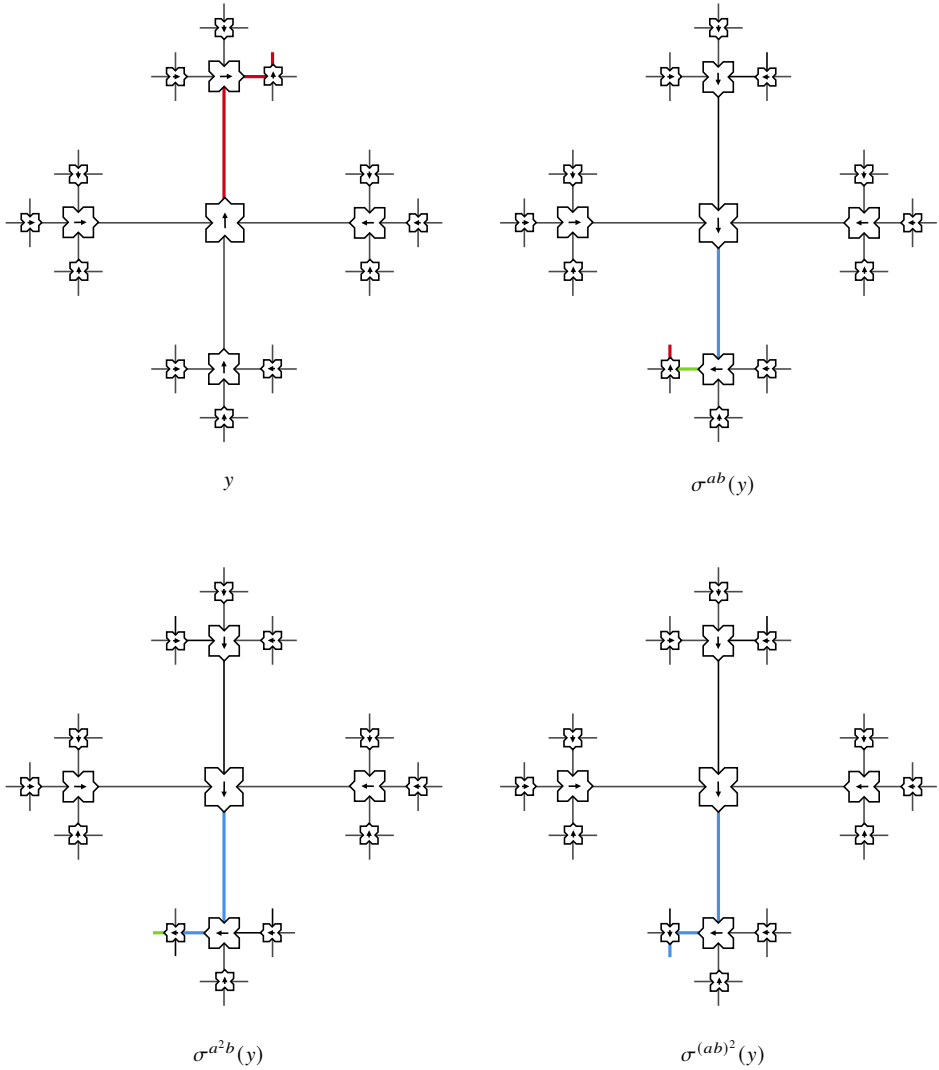

\begin{center}
    \includestandalone[scale=0.8]{ fig/minimal}  
    \caption{The first three steps to go from the word $(ba)^{\infty}$ to $(b^{-1}a^{-1})^{\infty}$, with error term $e_n = \ a^{-1}$ for all $n$. The original word is marked in red, the new one in blue, and the error term  in green.}
\end{center}
\end{figure}

\begin{lemma}
\label{minflow}
Let $y,y'\in Y_f$ be two configurations defined by the words $W$ and $W'$ respectively. Then, there exists a sequence $\{g_n\}_{n\in\N}$ in $\F_n\times\Z$ such that
$$\lim_{n\to\infty}\sigma^{g_n^{-1}}(y') = y,$$

and $|g_{n+1}| = |g_n| + 1$ for all $n\in\N$.
\end{lemma}

\begin{proof}
We would like to find $g_n$ such that $W(\sigma^{g_n^{-1}}(y')) = W_0\ ... \ W_n e_n W'$, with $|e_n| = 1$. We add the error term so we avoid forbidden flow patterns (we must avoid $W_n = (W'_0)^{-1}$ at every point) and for the size of the new word to increase by exactly 1 at each step. This term will disappear upon taking the limit.\\

We begin by introducing the directions involved in the error term:

$$a_{i} = \begin{cases} s_i^{-1} \ \text{ if }\ W'_0 = s_i \\ s_i\ \text{ if not }\end{cases}.$$
Then, we define $g_0 = a_iW_0^{-1}$ for $W_0\in (S\cup S^{-1})\setminus\{s_i, s_i^{-1}\}$. This way, we arrive at $W(\sigma^{g_0^{-1}}(y')) = W_0e_0W'$, where $e_0$ is the arrow we added as padding to avoid $W_0$ conflicting with $W'_0$. Next, we recursively define $g_n = a_i(W_0 \ ... \ W_n)^{-1}$ for $W_n\in(S\cup S^{-1})\setminus\{s_i, s_i^{-1}\}$. This way, we have
$$W(\sigma^{g_n^{-1}}(y')) = W_0 \ ...\ W_n e_n W',$$
where $e_n$ is the error term of size 1. Therefore,
$$\lim_{n\to\infty}\sigma^{g_n^{-1}}(y') = y.$$
\end{proof}

\begin{theorem}
There exists a minimal strongly aperiodic SFT on $\F_n\times\Z$.
\end{theorem}

\begin{proof}
Let us take two configurations $x'\otimes y'$ and $x\otimes y$ in $Z$. Because $X$ is minimal, there exists a sequence $\{(i_n, j_n)\}_{n\in\N}$ in $\Z^2$ with $(j_n)_{n\in\N}$ increasing, such that
$$\lim_{n\to\infty}\sigma^{(i_n,j_n)}(x') = x.$$
Let $\{g_n\}_{n\in\N}$ be the sequence from Lemma \ref{minflow}, that is,
$$\lim_{n\to\infty}\sigma^{g_n^{-1}}(y') = y.$$

Let $M\in\N$ be such that $j_M\geq 2$. Next, let $\{n_k\}_{k\geq M}$ be the increasing subsequence satisfying $n_k + 1 = j_{k}$. Then,

$$\sigma^{(g_{n_k}t^{i_k})^{-1}}(x'\otimes y') = (x'_{(-i_k,-j_k)}, W_0).$$

It follows that, \
$$\sigma^{(g_{n_k}t^{i_k})^{-1}}(x'\otimes y') = \sigma^{(i_k,j_k)}(x')\otimes \sigma^{g_{n_k}^{-1}}(y'),$$

and thus,
$$\lim_{k\to\infty}\sigma^{(g_{n_k}t^{i_k})^{-1}}(x'\otimes y') = x\otimes y.$$
This shows that $Z$ is minimal.
We conclude 
\end{proof}

As a consequence, because unimodular groups contain $\F_n\times\Z$ as a finite index normal subgroup, Proposition \ref{prop:normal_finite_index_lift} tells us that they admit minimal strongly aperiodic SFTs. In particular, both torus knot groups and $BS(n,n)$ admit this kind of subshift, as they are unimodular. This latter result improves~\cite{EsnayMoutot2020}.

\begin{cor}
Unimodular GBS groups admit minimal strongly aperiodic SFTs. In particular, both $\Lambda(n,m)$ and $BS(n,n)$ admit minimal strongly aperiodic SFTs.
\end{cor}

\section{Adaptation to the Baumslag-Solitar group $BS(2,3)$}\label{section:BS23}

Amenable Baumslag-Solitar groups $BS(1,n)$ are known to have strongly aperiodic SFTs~\cite{EsnayMoutot2020} and even minimal strongly aperiodic SFTs~\cite{aubrun2020tilings}. The case of $BS(m,n)$ for $m\neq n$ and $m,n>1$ has remained unsolved until now. Since all these groups are quasi-isometric~\cite{whyte2004large}, it is enough to focus on $BS(2,3)$. A weakly aperiodic SFT is known to exist on this group~\cite{AubrunKari2013} and we prove here that thanks to the path-folding technique, this construction can be modified to get strong aperiodicity. In a few words, the weakly aperiodic SFT relies on an embedding of $BS(2,3)$ into $\mathbb{R}^2$ that fails to be injective, and this injectivity default irremediably produces some periods in the SFT. We modify the embedding so that it now depends on an infinite path in the group, such that the choice of the path allows to break the existing periods.

\subsection{The group $BS(2,3)$}

\begin{figure}[!ht]
    \centering
    \begin{tikzpicture}[scale=0.33]
\foreach \x in {0,...,5} {
\draw[thick] (6*\x-1.5,0) -- (6*\x,4);
\draw[thick] (6*\x,4) to [controls=+(45:1) and +(135:1)] (6*\x+3,4);
\draw[thick] (6*\x+3,4) to [controls=+(45:1) and +(135:1)] (6*\x+6,4);
\draw[thick] (6*\x-1.5,0) to [controls=+(45:0.7) and +(135:0.7)] (6*\x+2-1.5,0);
\draw[thick] (6*\x+2-1.5,0) to [controls=+(45:0.7) and +(135:0.7)] (6*\x+4-1.5,0);
\draw[thick] (6*\x+4-1.5,0) to [controls=+(45:0.7) and +(135:0.7)] (6*\x+6-1.5,0);
\draw[thick] (6*\x+6,4) -- (6*\x+6-1.5,0);
}

\foreach \x in {0,...,4} {
\begin{scope}[shift={(3,0)},color=black!50]
\draw[thick] (6*\x+1.5,0-1.5) -- (6*\x,4);
\draw[thick] (6*\x,4) to [controls=+(45:1) and +(135:1)] (6*\x+3,4);
\draw[thick] (6*\x+3,4) to [controls=+(45:1) and +(135:1)] (6*\x+6,4);
\draw[thick] (6*\x+1.5,0-1.5) to [controls=+(45:0.7) and +(135:0.7)] (6*\x+2+1.5,0-1.5);
\draw[thick] (6*\x+2+1.5,0-1.5) to [controls=+(45:0.7) and +(135:0.7)] (6*\x+4+1.5,0-1.5);
\draw[thick] (6*\x+4+1.5,0-1.5) to [controls=+(45:0.7) and +(135:0.7)] (6*\x+6+1.5,0-1.5);
\draw[thick] (6*\x+6,4) -- (6*\x+6+1.5,0-1.5);
\end{scope}
}

\foreach \x in {0,1,2} {
\begin{scope}[scale=3/2,shift={(2,4*2/3)}]
\draw[thick] (6*\x,0) -- (6*\x,4);
\draw[thick] (6*\x,4) to [controls=+(45:1) and +(135:1)] (6*\x+3,4);
\draw[thick] (6*\x+3,4) to [controls=+(45:1) and +(135:1)] (6*\x+6,4);
\draw[thick] (6*\x,0) to [controls=+(45:0.7) and +(135:0.7)] (6*\x+2,0);
\draw[thick] (6*\x+2,0) to [controls=+(45:0.7) and +(135:0.7)] (6*\x+4,0);
\draw[thick] (6*\x+4,0) to [controls=+(45:0.7) and +(135:0.7)] (6*\x+6,0);
\draw[thick] (6*\x+6,4) -- (6*\x+6,0);
\end{scope}
}

\foreach \x in {0,1,2,3} {
\begin{scope}[scale=3/2,shift={(0,4*2/3)},color=black!50]
\draw[thick] (6*\x,0) -- (6*\x-1.5,4-1.5);
\draw[thick] (6*\x-1.5,4-1.5) to [controls=+(45:1) and +(135:1)] (6*\x+3-1.5,4-1.5);
\draw[thick] (6*\x+3-1.5,4-1.5) to [controls=+(45:1) and +(135:1)] (6*\x+6-1.5,4-1.5);
\draw[thick] (6*\x,0) to [controls=+(45:0.7) and +(135:0.7)] (6*\x+2,0);
\draw[thick] (6*\x+2,0) to [controls=+(45:0.7) and +(135:0.7)] (6*\x+4,0);
\draw[thick] (6*\x+4,0) to [controls=+(45:0.7) and +(135:0.7)] (6*\x+6,0);
\draw[thick] (6*\x+6-1.5,4-1.5) -- (6*\x+6,0);
\end{scope}
}

\foreach \x in {0,1,2} {
\begin{scope}[scale=3/2,shift={(4,4*2/3)},color=black!75]
\draw[thick] (6*\x,0) -- (6*\x+1.5,4-2.5);
\draw[thick] (6*\x+1.5,4-2.5) to [controls=+(45:1) and +(135:1)] (6*\x+3+1.5,4-2.5);
\draw[thick] (6*\x+3+1.5,4-2.5) to [controls=+(45:1) and +(135:1)] (6*\x+6+1.5,4-2.5);
\draw[thick] (6*\x,0) to [controls=+(45:0.7) and +(135:0.7)] (6*\x+2,0);
\draw[thick] (6*\x+2,0) to [controls=+(45:0.7) and +(135:0.7)] (6*\x+4,0);
\draw[thick] (6*\x+4,0) to [controls=+(45:0.7) and +(135:0.7)] (6*\x+6,0);
\draw[thick] (6*\x+6+1.5,4-2.5) -- (6*\x+6,0);
\end{scope}
}

\begin{scope}[decoration={
    markings,
    mark=at position 0.55 with {\arrow[scale=1.5]{latex}}}
    ]
\draw[postaction={decorate}] (0,4) -- (-1.5,0);
\draw node at (-1.5,2) {$t$};
\draw[postaction={decorate}] (0,4) to [controls=+(45:1) and +(135:1)] (3,4);
\draw node at (1.5,5.5) {$a$};
\draw[postaction={decorate}] (3,4) to [controls=+(45:1) and +(135:1)] (6,4);
\draw node at (4.5,5.5) {$a$};
\draw[postaction={decorate}] (6,4) -- (4.5,0);
\draw node at (4.5,2) {$t$};
\end{scope}

\begin{scope}[decoration={
    markings,
    mark=at position 0.65 with {\arrow[scale=1.5]{latex}}}
    ]
\draw[postaction={decorate}] (-1.5,0) to [controls=+(45:0.7) and +(135:0.7)] (0.5,0);
\draw[postaction={decorate}] (0.5,0) to [controls=+(45:0.7) and +(135:0.7)] (2.5,0);
\draw[postaction={decorate}] (2.5,0) to [controls=+(45:0.7) and +(135:0.7)] (4.5,0);
\draw node at (-0.5,-0.5) {$a$};
\draw node at (1.5,-0.5) {$a$};
\draw node at (3.5,-0.5) {$a$};
\end{scope}

\end{tikzpicture}    
    \caption{The Cayley graph of $BS(2,3)=\langle a,t \mid t^{-1}a^2t=a^3 \rangle$}
    \label{fig:Cayley_graph_BS23}
\end{figure}
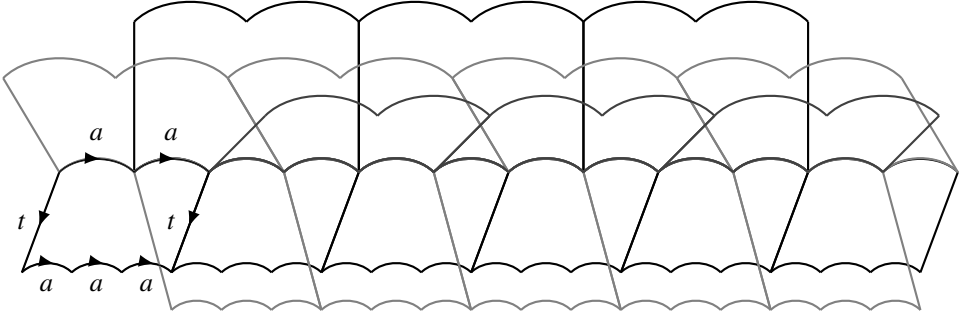

Since $BS(2,3)$ is an HNN-extension, by Britton's lemma we get a normal form for elements of $BS(2,3)$.

\begin{lemma}[(Normal form])\label{lemma:normal_form_BS23}
Every element $g\in BS(2,3)$ can be uniquely decomposed as $g = w a^k$, where $w$ is a freely reduced word over the alphabet $\{ t, at, t^{-1}, at^{-1}, a^2t^{-1}\}$ and $k\in\Z$.
\end{lemma}

\begin{proof}
Britton's lemma states that every element $g\in BS(2,3)$ has the form 
$$g = a^{N}t^{e_1}a^{m_1} \ ... \ t^{e_n}a^{m_n},$$
where $N\in \Z$ and $e_i=\pm 1$, such that if $e_i = 1$ then $m_i\in\{0,1\}$, if $e_i = -1$ then $m_i\in\{0,1,2\}$, and we never have a subword of the form $t^{\pm} a^0 t^{\mp}$.

Notice that if $g = a^{N}$, it is already in the form we are looking for. Next, if we have $g = a^{N}t$, we can decompose $N = 2d + r$, where $0\leq r < 2$ in order to change the order of the generators,
$$g = a^Nt = a^{2d + r}t = a^{r}t a^{3d}.$$

Analogously, if $g = a^{N}t^{-1}$, we decompose $N = 3d + r$ with $0\leq r < 3$ and arrive at
$$g = a^Nt^{-1} = a^{3d + r}t^{-1} = a^{r}t^{-1} a^{2d}.$$

Finally, for an arbitrary $g$, we simply iterate the two preceding procedures to arrive at an expression for $g$ in the sought after form.
\end{proof}

\subsection{A flow SFT on $BS(2,3)$}

Consider the alphabet $A=\{ t, at, t^{-1}, at^{-1}, a^2t^{-1}\}$ and the SFT $Y\subset A^{BS(2,3)}$ defined by the following local rules: for every group element $g\in BS(2,3)$ and every configuration $y\in Y$,
\begin{itemize}
    \item $y_g=y_{g\cdot a^2}$ if $y_g\in\{t,at\}$
    \item $y_g=y_{g\cdot a^3}$ if $y_g\in\{t^{-1},at^{-1}, a^{2}t^{-1}\}$
    \item if $y_g=u\in A$ then for every $v\in A\setminus\{u\}$ we have $y_{g\cdot v^{-1}}=v$
\end{itemize}

This SFT can be equivalently and in a more visual way defined through the finite patterns with support $\{1,a,a^2,t, ta,ta^2,ta^3\}\cup A$ pictured below. 
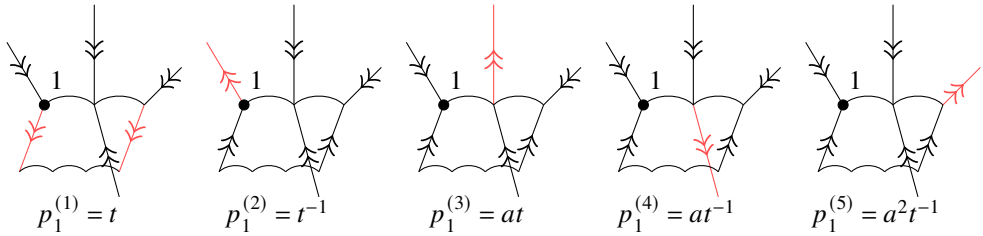
\begin{figure}[!ht]
    \begin{bigcenter}
    \begin{tikzpicture}[scale=0.22]

\begin{scope}[decoration={
    markings,
    mark=at position 0.55 with {\arrow[scale=2]{>>}}}
    ]
\draw[postaction={decorate}] (4.5,0-1.5) -- (3,4);
\begin{scope}[scale=3/2,shift={(2,4*2/3)}]
\draw[postaction={decorate}] (0,4) -- (0,0);
\end{scope}
\begin{scope}[scale=3/2,shift={(0,4*2/3)}]
\draw[postaction={decorate}] (-1.5,4-1.5) -- (0,0);
\end{scope}
\begin{scope}[scale=3/2,shift={(4,4*2/3)}]
\draw[postaction={decorate}] (1.5,4-2.5) -- (0,0);
\end{scope}
\draw[postaction={decorate}, color=rouge] (0,4) -- (-1.5,0);
\draw[postaction={decorate}, color=rouge] (6,4) -- (4.5,0);
\end{scope}
\draw[] (3,4) to [controls=+(45:1) and +(135:1)] (6,4);
\draw[] (0,4) to [controls=+(45:1) and +(135:1)] (3,4);
\draw[] (-1.5,0) to [controls=+(45:0.7) and +(135:0.7)] (0.5,0);
\draw[] (0.5,0) to [controls=+(45:0.7) and +(135:0.7)] (2.5,0);
\draw[] (2.5,0) to [controls=+(45:0.7) and +(135:0.7)] (4.5,0);
\filldraw (0,4) circle (0.3cm);
\draw (0.75,5.5) node {$1$};
\draw (2,-2.5) node {$p^{(1)}_1=t$};

\begin{scope}[shift={(12,0)}]
\begin{scope}[decoration={
    markings,
    mark=at position 0.55 with {\arrow[scale=2]{>>}}}
    ]
\draw[postaction={decorate}] (4.5,0-1.5) -- (3,4);
\begin{scope}[scale=3/2,shift={(2,4*2/3)}]
\draw[postaction={decorate}] (0,4) -- (0,0);
\end{scope}
\begin{scope}[scale=3/2,shift={(0,4*2/3)}]
\draw[postaction={decorate}, color=rouge] (0,0) -- (-1.5,4-1.5);
\end{scope}
\begin{scope}[scale=3/2,shift={(4,4*2/3)}]
\draw[postaction={decorate}] (1.5,4-2.5) -- (0,0);
\end{scope}
\draw[postaction={decorate}] (-1.5,0) -- (0,4);
\draw[postaction={decorate}] (4.5,0) -- (6,4);
\end{scope}
\draw[] (3,4) to [controls=+(45:1) and +(135:1)] (6,4);
\draw[] (0,4) to [controls=+(45:1) and +(135:1)] (3,4);
\draw[] (-1.5,0) to [controls=+(45:0.7) and +(135:0.7)] (0.5,0);
\draw[] (0.5,0) to [controls=+(45:0.7) and +(135:0.7)] (2.5,0);
\draw[] (2.5,0) to [controls=+(45:0.7) and +(135:0.7)] (4.5,0);
\filldraw (0,4) circle (0.3cm);
\draw (0.75,5.5) node {$1$};
\draw (2,-2.5) node {$p^{(2)}_1=t^{-1}$};
\end{scope}

\begin{scope}[shift={(24,0)}]
\begin{scope}[decoration={
    markings,
    mark=at position 0.55 with {\arrow[scale=2]{>>}}}
    ]
\draw[postaction={decorate}] (4.5,0-1.5) -- (3,4);
\begin{scope}[scale=3/2,shift={(2,4*2/3)}]
\draw[postaction={decorate}, color=rouge] (0,0) -- (0,4);
\end{scope}
\begin{scope}[scale=3/2,shift={(0,4*2/3)}]
\draw[postaction={decorate}] (-1.5,4-1.5) -- (0,0);
\end{scope}
\begin{scope}[scale=3/2,shift={(4,4*2/3)}]
\draw[postaction={decorate}] (1.5,4-2.5) -- (0,0);
\end{scope}
\draw[postaction={decorate}] (-1.5,0) -- (0,4);
\draw[postaction={decorate}] (4.5,0) -- (6,4);
\end{scope}
\draw[] (3,4) to [controls=+(45:1) and +(135:1)] (6,4);
\draw[] (0,4) to [controls=+(45:1) and +(135:1)] (3,4);
\draw[] (-1.5,0) to [controls=+(45:0.7) and +(135:0.7)] (0.5,0);
\draw[] (0.5,0) to [controls=+(45:0.7) and +(135:0.7)] (2.5,0);
\draw[] (2.5,0) to [controls=+(45:0.7) and +(135:0.7)] (4.5,0);
\filldraw (0,4) circle (0.3cm);
\draw (0.75,5.5) node {$1$};
\draw (2,-2.5) node {$p^{(3)}_1=at$};
\end{scope}

\begin{scope}[shift={(36,0)}]
\begin{scope}[decoration={
    markings,
    mark=at position 0.55 with {\arrow[scale=2]{>>}}}
    ]
\draw[postaction={decorate}, color=rouge] (3,4) -- (4.5,0-1.5);
\begin{scope}[scale=3/2,shift={(2,4*2/3)}]
\draw[postaction={decorate}] (0,4) -- (0,0);
\end{scope}
\begin{scope}[scale=3/2,shift={(0,4*2/3)}]
\draw[postaction={decorate}] (-1.5,4-1.5) -- (0,0);
\end{scope}
\begin{scope}[scale=3/2,shift={(4,4*2/3)}]
\draw[postaction={decorate}] (1.5,4-2.5) -- (0,0);
\end{scope}
\draw[postaction={decorate}] (-1.5,0) -- (0,4);
\draw[postaction={decorate}] (4.5,0) -- (6,4);
\end{scope}
\draw[] (3,4) to [controls=+(45:1) and +(135:1)] (6,4);
\draw[] (0,4) to [controls=+(45:1) and +(135:1)] (3,4);
\draw[] (-1.5,0) to [controls=+(45:0.7) and +(135:0.7)] (0.5,0);
\draw[] (0.5,0) to [controls=+(45:0.7) and +(135:0.7)] (2.5,0);
\draw[] (2.5,0) to [controls=+(45:0.7) and +(135:0.7)] (4.5,0);
\filldraw (0,4) circle (0.3cm);
\draw (0.75,5.5) node {$1$};
\draw (2,-2.5) node {$p^{(4)}_1=at^{-1}$};
\end{scope}

\begin{scope}[shift={(48,0)}]
\begin{scope}[decoration={
    markings,
    mark=at position 0.55 with {\arrow[scale=2]{>>}}}
    ]
\draw[postaction={decorate}] (4.5,0-1.5) -- (3,4);
\begin{scope}[scale=3/2,shift={(2,4*2/3)}]
\draw[postaction={decorate}] (0,4) -- (0,0);
\end{scope}
\begin{scope}[scale=3/2,shift={(0,4*2/3)}]
\draw[postaction={decorate}] (-1.5,4-1.5) -- (0,0);
\end{scope}
\begin{scope}[scale=3/2,shift={(4,4*2/3)}]
\draw[postaction={decorate}, color=rouge] (0,0) -- (1.5,4-2.5);
\end{scope}
\draw[postaction={decorate}] (-1.5,0) -- (0,4);
\draw[postaction={decorate}] (4.5,0) -- (6,4);
\end{scope}
\draw[] (3,4) to [controls=+(45:1) and +(135:1)] (6,4);
\draw[] (0,4) to [controls=+(45:1) and +(135:1)] (3,4);
\draw[] (-1.5,0) to [controls=+(45:0.7) and +(135:0.7)] (0.5,0);
\draw[] (0.5,0) to [controls=+(45:0.7) and +(135:0.7)] (2.5,0);
\draw[] (2.5,0) to [controls=+(45:0.7) and +(135:0.7)] (4.5,0);
\filldraw (0,4) circle (0.3cm);
\draw (0.75,5.5) node {$1$};
\draw (2,-2.5) node {$p^{(5)}_1=a^2t^{-1}$};
\end{scope}

\end{tikzpicture}   
    \end{bigcenter}
    \caption{The allowed patterns $p^{(1)}$ to $p^{(5)}$ for the flow SFT on $BS(2,3)$. For more readability the outgoing edges are pictured in red.}
    \label{fig:flow_patterns_BS}
\end{figure}

We denote by $Y_\text{flow}$ the resulting SFT on $BS(2,3)$. Notice that for each flow configuration $y\in Y_\text{flow}$ and for every $g\in BS(2,3)$, the restriction of $y$ to $g\cdot a^\Z$ is necessarily periodic, with this period being either $a^2$ or $a^3$. More precisely the coset $g\cdot a^\Z$ is $a^2$-periodic if $y_g\in \{t,at\}$ and $a^3$-periodic if $y_g\in \{t^{-1},at^{-1},a^2t^{-1}\}$. Consequently we may represent $y$ just by a flow on the Bass-Serre tree of $BS(2,3)$, that is to say an edge coloring of the complete tree of degree~$5$ where each vertex has a single outgoing arrow and four incoming arrows.

\medskip

In the same fashion as in Section~\ref{sec:flow}, we can express flow configurations from $Y_\text{flow}$ as infinite words. 

\begin{prop}
There is a bijective corresponding between configurations of $Y_\text{flow}$ and semi-infinite words on the alphabet $A=\{ t, at, t^{-1}, at^{-1}, a^2t^{-1}\}$.
\end{prop}

\begin{proof}
If $y\in Y_\text{flow}$ we denote by $W(y)$ the word in $A^\N$ given by the recursion starting with $W_0 = y_1$ and proceeding with $W_n = y_{W_0 ... W_{n-1}}$.

Reciprocally if $W$ is a word in $A^\N$ we define a flow configuration $Y_\text{flow}$. We begin by defining $y_{W_0 \ ... \ W_{n-1}} = W_n$ for all $n\geq 0$. Next, we can determine all other values through the use of the periodicity of $a$-cosets and the third rule rule defining the flow SFT, as shown by the allowed patterns in Figure \ref{fig:flow_patterns_BS}.
\end{proof}

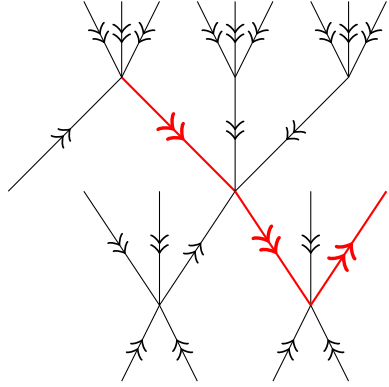
\begin{figure}[!ht]
    \centering
    \begin{tikzpicture}[scale=0.5]

\coordinate (O) at (0,0);
\coordinate (T) at (-3,3);
\coordinate (aT) at (0,3);
\coordinate (aaT) at (3,3);
\coordinate (t) at (-2,-3);
\coordinate (at) at (2,-3);
\coordinate (tt) at (-3,-5);
\coordinate (tat) at (-1,-5);

\coordinate (taT) at (-2,0);
\coordinate (taaT) at (-4,0);

\coordinate (ataT) at (2,0);
\coordinate (ataaT) at (4,0);

\coordinate (at) at (2,-3);
\coordinate (att) at (1,-5);
\coordinate (atat) at (3,-5);

\coordinate (TT) at (-4,5);
\coordinate (TaT) at (-3,5);
\coordinate (TaaT) at (-2,5);

\coordinate (aTT) at (-1,5);
\coordinate (aTaT) at (0,5);
\coordinate (aTaaT) at (1,5);

\coordinate (aaTT) at (2,5);
\coordinate (aaTaT) at (3,5);
\coordinate (aaTaaT) at (4,5);

\coordinate (Tat) at (-6,0);

\begin{scope}[decoration={
    markings,
    mark=at position 0.55 with {\arrow[scale=2]{>>}}}
    ]

\draw[postaction={decorate}] (taT) -- (t); 
\draw[postaction={decorate}] (taaT) -- (t); 
\draw[postaction={decorate}] (tat) -- (t); 
\draw[postaction={decorate}] (tt) -- (t); 
\draw[postaction={decorate}] (atat) -- (at); 
\draw[postaction={decorate}] (att) -- (at);

\draw[postaction={decorate}] (ataT) -- (at); 
\draw[postaction={decorate}, thick, color=red] (at) -- (ataaT); 

\draw[postaction={decorate}, thick, color=red] (T) -- (O); 
\draw[postaction={decorate}] (aT) -- (O); 
\draw[postaction={decorate}] (aaT) -- (O); 
\draw[postaction={decorate}] (t) -- (O); 
\draw[postaction={decorate}, thick, color=red] (O) -- (at); 

\draw[postaction={decorate}] (TT) -- (T); 
\draw[postaction={decorate}] (TaT) -- (T);
\draw[postaction={decorate}] (TaaT) -- (T); 
\draw[postaction={decorate}] (Tat) -- (T);

\draw[postaction={decorate}] (aTT) -- (aT); 
\draw[postaction={decorate}] (aTaT) -- (aT);
\draw[postaction={decorate}] (aTaaT) -- (aT); 

\draw[postaction={decorate}] (aaTT) -- (aaT); 
\draw[postaction={decorate}] (aaTaT) -- (aaT);
\draw[postaction={decorate}] (aaTaaT) -- (aaT);

\end{scope}

\end{tikzpicture}    
    \caption{A configuration in the flow SFT on $BS(2,3)$, pictured on the Bass-Serre tree only.}
    \label{fig:config_flot_BS23}
\end{figure}

\begin{prop}\label{prop:flow_follows_period_BS}
If $y\in Y_\text{flow}$ has period $g\in BS(2,3)$ with normal form decomposition $g^{-1} = wa^k$, then $W(y)$ is either the infinite word $w^\N$ or the infinite word $(w^{-1})^\N$.
\end{prop}

\begin{proof}

The proof is analogous to the proof of Proposition~\ref{prop:flow_period}.
\end{proof}

\subsection{Embedding $BS(2,3)$ into $\mathbb{R}^2$ along a flow configuration}

We also define an embedding of $BS(2,3)$ into $\mathbb{R}^2$ driven  by a flow configuration $y\in Y_\text{flow}$, denoted by $\Phi_y:BS(2,3)\to\mathbb{R}^2$, that is first recursively defined on finite words $w$ on the alphabet $B=\{ a,t,a^{-1},t^{-1}\}$. 

We define $\Phi_y(w)=\left( \alpha(w),\beta_y(w)\right)$ recursively, coordinate by coordinate. Let $\varepsilon$ denote the empty word. The second coordinate is such that, for $u\in \left\{ t, at, a^2t, t^{-1}, at^{-1}\right\}$
        \begin{align*}
        \beta_y(\varepsilon) &= 0\\
        \beta_y(w.u) &=\beta_y(w)+1\text{ if }y_w=u\\
         &=\beta_y(w)-1\text{ otherwise}\\
        \beta_y(w.a) &= \beta_y(w) = \beta_y(w.a^{-1}).
        \end{align*}
This coordinate represents the height of a group element according to the flow configuration~$y$. The first coordinate is:
        \begin{align*}
        \alpha(\varepsilon) &= 0\\
        \alpha(w.t) &=\alpha(w.t^{-1})=\alpha(w)\\
        \alpha(w.a) &=\alpha(w)+{\left(\frac{2}{3}\right)}^{\beta(w)}\\
        \alpha(w.a^{-1}) &=\alpha(w)-{\left(\frac{2}{3}\right)}^{\beta(w)}.
        \end{align*}
where $\beta(w):=\parallel w\parallel_t = |w|_t - |w|_{t^{-1}}$ counts the contribution of the generator $t$ to $w$. This first coordinate $\alpha(w)$ is exactly the first coordinate of the $\Phi$ embedding given in~\cite{AubrunKari2021}. The difference with $\Phi_y$ lies in the second coordinate, $\beta_y(w)$, that no longer follows the generator $t$ but the path induced by the flow configuration $y$ instead.

\begin{prop}\label{prop:beta_y_well_defined}
For every $g\in BS(2,3)$ the value of $\beta_y(w)$ does not depend on the choice for the word $w$ that represents $g$, hence $\beta_y$ is well-defined on $BS(2,3)$.
\end{prop}

\begin{proof}
We prove this by induction on the size of the normal form of Lemma~\ref{lemma:normal_form_BS23}. Since $\beta_y(w.a^{\pm1}) = \beta_y(w)$ we can get rid of the last term, $a^k$, in the writing of the normal form; it does not contribute to $\beta_y$.
Assume every $\beta_y$ is well-defined for all group elements that can be written with $n$ letters from alphabet $A=\left\{ t, at, t^{-1}, at^{-1}, a^2t^{-1}\right\}$. Let $g\in BS(2,3)$ be an element with normal form $w\in A^{n+1}$. Denote $g'$ the group element with normal form $w_0\dots w_{n-1}\in A^n$ . Then
\begin{align*}
    \beta_y(g) &= \beta_y(w_0\dots w_n).
\end{align*}
There are two cases, depending on whether $w_n=y_{g'}$ or $w_n\neq y_{g'}$. In the first case
\begin{align*}
    \beta_y(g) &= \beta_y(w_0\dots w_n) +1\\
    &= \beta_y(g') +1\text{ by induction hypothesis},
\end{align*}
and in the second case
\begin{align*}
    \beta_y(g) &= \beta_y(w_0\dots w_n) -1\\
    &= \beta_y(g') -1\text{ by induction hypothesis},
\end{align*}
so that $\beta_y(g)$ does not depend on the chosen word.
\end{proof}

Proofs of analogous results for $\alpha$ and $\beta$ can be performed in a similar way. Again following~\cite{AubrunKari2021} we define $\lambda:BS(2,3)\to\mathbb{R}$ as
\[
\lambda(g) = \frac{1}{2}\left(\frac{3}{2}\right)^{\beta(g)}\alpha(g).
\]

\begin{prop}\label{prop:lambda_useful_eq}
Let $g$ be an element of $BS(2,3)$. Then for $i=0, \dots, 2$:
\begin{enumerate}
    \item $\beta(g\cdot ta^i) = \beta(g)+1$;
    \item $\lambda(g\cdot ta^i) = \frac{3}{2}\lambda(g)+ \frac{i}{2}$.
\end{enumerate}
\end{prop}

\begin{proof}
The first point is a direct application of the rules that define $\beta$. For the second point we have that
\begin{align*}
    \lambda(g\cdot ta^i) &= \frac{1}{2}\left(\frac{3}{2}\right)^{\beta(g\cdot ta^i)}\alpha(g\cdot ta^i)\\ 
    &= \frac{1}{2}\left(\frac{3}{2}\right)^{\beta(g)+1}\alpha(g\cdot ta^i)\\
    &= \frac{1}{2}\left(\frac{3}{2}\right)^{\beta(g)+1}\left(\alpha(g)+ i \cdot \left(\frac{2}{3}\right)^{\beta(g\cdot t)}\right)\\
    &= \frac{3}{2}\lambda(g) + \frac{i}{2}\left(\frac{3}{2}\right)^{\beta(g)+1}\left(\frac{2}{3}\right)^{\beta(g)+1}\\
    \lambda(g\cdot ta^i) &= \frac{3}{2}\lambda(g)+ \frac{i}{2}.
\end{align*}
\end{proof}

\subsection{A strongly aperiodic SFT on $BS(2,3)$}

To construct an aperiodic SFT we will add a new layer of tiles to the flow shift. These new tiles are Wang tiles for $BS(2,3)$ that encode a piecewise linear function.

\begin{figure}[!ht]
    \centering
    \begin{tikzpicture}[scale=0.6]

\draw[thick] (0,0) -- (0,4)-- (0,4) to [controls=+(45:1) and +(135:1)] (3,4) -- (3,4) to [controls=+(45:1) and +(135:1)] (6,4) -- (6,4) -- (6,0) -- (6,0) to [controls=+(135:0.7) and +(45:0.7)] (4,0) -- (4,0) to [controls=+(135:0.7) and +(45:0.7)] (2,0) -- (2,0) to [controls=+(135:0.7) and +(45:0.7)] (0,0) -- cycle ;

\draw (1.5,5) node{$t_1$};
\draw (4.5,5) node{$t_2$};
\draw (-1.25,2) node{$\ell$};
\draw (7.25,2) node{$r$};
\draw (1,0.8) node{$b_1$};
\draw (3,0.8) node{$b_2$};
\draw (5,0.8) node{$b_3$};

\end{tikzpicture}    
    \caption{A Wang tile for $BS(2,3)$}
    \label{fig:wangtile_BS23}
\end{figure}
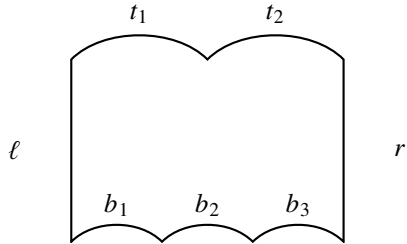

Each tile consists 7-tuple of integers $s = (t_1, t_2, l, b_1, b_2, b_3, r)$, as shown in Figure \ref{fig:wangtile_BS23}. Let $\tau$ be a set of these Wang tiles. We say a that a configuration $z\in\tau^{BS(2,3)}$ is a valid tiling if the colors of neighboring tiles match. More explicitly,for every $g\in BS(2,3)$ we must have:\label{eq:matching_rules}
\begin{align*}
z_{g}(r) &= z_{g\cdot a^2}(\ell)\\
z_{g}(b_i) &= z_{g\cdot a^{i-1}t}(t_1)\text{ for }i=1,2,3\\
z_{g}(b_i) &= z_{g\cdot a^{i-2}t}(t_2)\text{ for }i=1,2,3.
\end{align*}

We say that a Wang tile for $BS(2,3)$ computes a function $f:I\subset\mathbb{R}\to I$ if
\[
f\left( \frac{t_1+t_2}{2}\right) + \ell = \frac{b_1+b_2+b_3}{3}+r.
\]
If this equality holds for $f$, we say that the tile computes $f$ along the generator $t$. If $f$ is invertible and the tiles computes $f^{-1}$, we say that the tiles computes $f$ against the generator $t$.

Let us define the circle $I = \left[\frac{1}{10};\frac{5}{2}\right]/_{\frac{1}{10}\sim\frac{5}{2}}$. We introduce $T:I \to I$ the piecewise linear map defined by

\begin{tabular}{cc}
\begin{minipage}{0.59\linewidth}
\[
T: x \mapsto \left\{
  \begin{array}{c}
  \frac{5}{2}x \ \text{ if } \ x\in[\frac{1}{10};1]\\
  \\
  \frac{1}{10}x \ \text{ if }\ x\in ]1;\frac{5}{2}[
  \end{array}
  \right.
\]
\end{minipage}
&
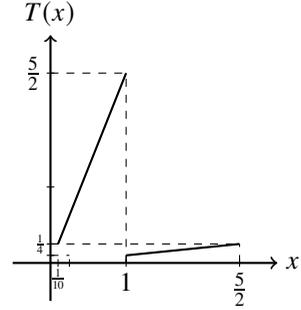
\begin{minipage}{0.39\linewidth}

    \begin{tikzpicture}

\draw[thick,->] (-0.5,0) -> (3,0);
\draw (0,3) node[above]{$T(x)$};
\draw (1,0) node[below]{$1$};
\draw (1/10,0) node[below,scale=0.6]{$\frac{1}{10}$};
\draw (5/2,0) node[below]{$\frac{5}{2}$};
\draw[thick,->] (0,-0.5) -> (0,3);
\draw (3,0) node[right]{$x$};
\draw (0,1/4) node[left,scale=0.6]{$\frac{1}{4}$};
\draw (0,5/2) node[left]{$\frac{5}{2}$};
\foreach \x in {1,1/4,1/10,5/2} {
\draw (\x,-0.05) -- (\x,0.05);
\draw (-0.05,\x) -- (0.05,\x);
}

\draw[thick] (1,1/10) -- (5/2,1/4);
\draw[thick] (1/10,1/4) -- (1,5/2);

\draw[dashed] (0,1/10) -- (1/4,1/10);
\draw[dashed] (0,1/4) -- (1/10,1/4);
\draw[dashed] (1,0) -- (1,1/10);
\draw[dashed] (1,0) -- (1,5/2);
\draw[dashed] (5/2,0) -- (5/2,1/4);
\draw[dashed] (0,5/2) -- (1,5/2);
\draw[dashed] (0,1/4) -- (5/2,1/4);

    
\end{tikzpicture}    
\end{minipage}
\end{tabular}

This linear map is invertible with inverse 
\[
T^{-1}: x \mapsto \left\{
  \begin{array}{c}
  10x \ \text{ if } \ x\in]\frac{1}{10};\frac{1}{4}[\\
  \\
  \frac{2}{5}x \ \text{ if } \ x\in [\frac{1}{4};\frac{5}{2}]
  \end{array}
  \right.
\]

It is not difficult to see that $T$ admits immortal points, i.e. reals numbers $x$ such that for every $k\in\Z$, $T^k(x)$ lies in $I$. It is also easy to check, since $5$ and $2$ are coprime, that $T$ is aperiodic, meaning that for every $x\in I$ if $T^k(x)=x$ for some integer $k\in\Z$, then $k=0$.

\medskip

We do not use the function used in~\cite{Kari1996} to construct a strongly aperiodic SFT on $\Z^2$ and in \cite{AubrunKari2013} to construct a weakly aperiodic SFT on $BS(3,2)$, because it may cause trouble in our construction. Indeed a careful observation of how tiles are built (see~\cite{AubrunKari21add} for the bounds on the values for $\ell$) shows that the tileset corresponding to the piece of the function given by $x\mapsto \frac{n}{m}x$ is empty! It is safer to use a piecewise linear function where no multiplicative coefficient matches $\frac{2}{3}$, hence our choice for $T$.

\medskip

Thanks to the machinery presented in~\cite{AubrunKari2013,AubrunKari2021}, we can define from the function $T$ two tilesets: first $\tau_T$ that computes $T$ along $t$ then $\tau_{T^{-1}}$ that computes $T^{-1}$ along $t$ --or equivalently computes $T$ against $t$. We thus define the following quantities that depend on three parameters: a function $f$, that can be either $T$ or $T^{-1}$ in our case, a real number $x\in\left[ \frac{1}{10};\frac{5}{2}\right]$ and a group element $g\in BS(2,3)$.

\begin{equation}
\label{eq:tilecolors}
\hspace*{-3mm}
\begin{array}{lcl}
         t_k(x,g) &=& \lfloor \left(2\lambda(g)+k \right)x \rfloor - \lfloor \left(2\lambda(g)+(k-1) \right)x \rfloor \text{~for }k=1,2\\ \\
         b_k(f,x,g) &=& \lfloor \left(3\lambda(g)+k \right)f(x) \rfloor - \lfloor \left(3\lambda(g)+(k-1) \right)f(x) \rfloor \text{~for }k=1,2,3\\ \\
         \ell(f,x,g) &=& \frac{1}{2}f\left(\lfloor 2\lambda(g)x \rfloor\right) - \frac{1}{3}\lfloor 3\lambda(g)f(x) \rfloor \\ \\
         r(f,x,g) &=& \frac{1}{2}f\left(\lfloor \left(2\lambda(g)+2\right)x \rfloor\right) - \frac{1}{3}\lfloor \left(3\lambda(g)+3\right)f(x) \rfloor 
\end{array}
\end{equation}

We gather $\tau_T$ and $\tau_{T^{-1}}$ into a single tileset $\tau$ and combine it with the flow SFT $Y_\text{flow}$ to define the SFT $Y_T$ over the alphabet $A\times \tau$ as follows: every configuration $z\in Y_T$, which we denote $z_g=(y_g,\tau_g)$ for every $g\in BS(2,3)$, satisfies:
\begin{itemize}
    \item if $y_g=t$ then $\tau_g\in \tau_T$;
    \item if $y_g\in \{at,t^{-1},t^{-1},at^{-1},a^2t^{-1}\}$ then $\tau_g\in \tau_{T^{-1}}$.
\end{itemize}
These two conditions impose that the computation of iterates of $T$ follows the flow: we put tiles that compute $T$ on outgoing arrows and tiles that computes $T^{-1}$ on incoming arrows (see Figure~\ref{fig:flow_patterns_BS_compute_T}).

\begin{figure}[!ht]
    \begin{bigcenter}
    \begin{tikzpicture}[scale=0.17]

\begin{scope}[decoration={
    markings,
    mark=at position 0.55 with {\arrow[scale=2]{>>}}}
    ]
\draw[postaction={decorate}] (4.5,0-1.5) -- (3,4);
\begin{scope}[scale=3/2,shift={(2,4*2/3)}]
\draw[postaction={decorate}] (0,4) -- (0,0);
\end{scope}
\begin{scope}[scale=3/2,shift={(0,4*2/3)}]
\draw[postaction={decorate}] (-1.5,4-1.5) -- (0,0);
\end{scope}
\begin{scope}[scale=3/2,shift={(4,4*2/3)}]
\draw[postaction={decorate}] (1.5,4-2.5) -- (0,0);
\end{scope}
\draw[postaction={decorate}, color=rouge] (0,4) -- (-1.5,0);
\draw[postaction={decorate}, color=rouge] (6,4) -- (4.5,0);
\end{scope}
\draw[] (3,4) to [controls=+(45:1) and +(135:1)] (6,4);
\draw[] (0,4) to [controls=+(45:1) and +(135:1)] (3,4);
\draw[] (-1.5,0) to [controls=+(45:0.7) and +(135:0.7)] (0.5,0);
\draw[] (0.5,0) to [controls=+(45:0.7) and +(135:0.7)] (2.5,0);
\draw[] (2.5,0) to [controls=+(45:0.7) and +(135:0.7)] (4.5,0);
\filldraw (0,4) circle (0.3cm);
\draw (-1.5,4) node {$x$};
\draw (0.5,5.5) node[scale=0.7]{$g$};
\draw (2,-6) node{$y_g=t$};
\draw (-3.5,0) node[scale=0.7]{$T(x)$};
\draw (-3.5,9) node[scale=0.7]{$T^{-1}(x)$};
\draw (3,11) node[scale=0.7]{$T^{-1}(x)$};
\draw (5,-3) node[scale=0.7]{$T^{-1}(x)$};
\draw (8,7) node[scale=0.7]{$T^{-1}(x)$};

\begin{scope}[shift={(18,0)}]
\begin{scope}[decoration={
    markings,
    mark=at position 0.55 with {\arrow[scale=2]{>>}}}
    ]
\draw[postaction={decorate}] (4.5,0-1.5) -- (3,4);
\begin{scope}[scale=3/2,shift={(2,4*2/3)}]
\draw[postaction={decorate}] (0,4) -- (0,0);
\end{scope}
\begin{scope}[scale=3/2,shift={(0,4*2/3)}]
\draw[postaction={decorate}, color=rouge] (0,0) -- (-1.5,4-1.5);
\end{scope}
\begin{scope}[scale=3/2,shift={(4,4*2/3)}]
\draw[postaction={decorate}] (1.5,4-2.5) -- (0,0);
\end{scope}
\draw[postaction={decorate}] (-1.5,0) -- (0,4);
\draw[postaction={decorate}] (4.5,0) -- (6,4);
\end{scope}
\draw[] (3,4) to [controls=+(45:1) and +(135:1)] (6,4);
\draw[] (0,4) to [controls=+(45:1) and +(135:1)] (3,4);
\draw[] (-1.5,0) to [controls=+(45:0.7) and +(135:0.7)] (0.5,0);
\draw[] (0.5,0) to [controls=+(45:0.7) and +(135:0.7)] (2.5,0);
\draw[] (2.5,0) to [controls=+(45:0.7) and +(135:0.7)] (4.5,0);
\filldraw (0,4) circle (0.3cm);
\draw (-1.5,4) node[scale=0.7]{$x$};
\draw (0.5,5.5) node[scale=0.7]{$g$};
\draw (2,-6) node{$y_g=t^{-1}$};
\draw (-4,0) node[scale=0.7]{$T^{-1}(x)$};
\draw (-3.5,9) node[scale=0.7]{$T(x)$};
\draw (3,11) node[scale=0.7]{$T^{-1}(x)$};
\draw (5,-3) node[scale=0.7]{$T^{-1}(x)$};
\draw (8,7) node[scale=0.7]{$T^{-1}(x)$};
\end{scope}

\begin{scope}[shift={(36,0)}]
\begin{scope}[decoration={
    markings,
    mark=at position 0.55 with {\arrow[scale=2]{>>}}}
    ]
\draw[postaction={decorate}] (4.5,0-1.5) -- (3,4);
\begin{scope}[scale=3/2,shift={(2,4*2/3)}]
\draw[postaction={decorate}, color=rouge] (0,0) -- (0,4);
\end{scope}
\begin{scope}[scale=3/2,shift={(0,4*2/3)}]
\draw[postaction={decorate}] (-1.5,4-1.5) -- (0,0);
\end{scope}
\begin{scope}[scale=3/2,shift={(4,4*2/3)}]
\draw[postaction={decorate}] (1.5,4-2.5) -- (0,0);
\end{scope}
\draw[postaction={decorate}] (-1.5,0) -- (0,4);
\draw[postaction={decorate}] (4.5,0) -- (6,4);
\end{scope}
\draw[] (3,4) to [controls=+(45:1) and +(135:1)] (6,4);
\draw[] (0,4) to [controls=+(45:1) and +(135:1)] (3,4);
\draw[] (-1.5,0) to [controls=+(45:0.7) and +(135:0.7)] (0.5,0);
\draw[] (0.5,0) to [controls=+(45:0.7) and +(135:0.7)] (2.5,0);
\draw[] (2.5,0) to [controls=+(45:0.7) and +(135:0.7)] (4.5,0);
\filldraw (0,4) circle (0.3cm);
\draw (-1.5,4) node[scale=0.7]{$x$};
\draw (0.5,5.5) node[scale=0.7]{$g$};
\draw (2,-6) node{$y_g=at$};
\draw (-4,0) node[scale=0.7]{$T^{-1}(x)$};
\draw (-3.5,9) node[scale=0.7]{$T^{-1}(x)$};
\draw (3,11) node[scale=0.7]{$T(x)$};
\draw (5,-3) node[scale=0.7]{$T^{-1}(x)$};
\draw (8,7) node[scale=0.7]{$T^{-1}(x)$};
\end{scope}

\begin{scope}[shift={(54,0)}]
\begin{scope}[decoration={
    markings,
    mark=at position 0.55 with {\arrow[scale=2]{>>}}}
    ]
\draw[postaction={decorate}, color=rouge] (3,4) -- (4.5,0-1.5);
\begin{scope}[scale=3/2,shift={(2,4*2/3)}]
\draw[postaction={decorate}] (0,4) -- (0,0);
\end{scope}
\begin{scope}[scale=3/2,shift={(0,4*2/3)}]
\draw[postaction={decorate}] (-1.5,4-1.5) -- (0,0);
\end{scope}
\begin{scope}[scale=3/2,shift={(4,4*2/3)}]
\draw[postaction={decorate}] (1.5,4-2.5) -- (0,0);
\end{scope}
\draw[postaction={decorate}] (-1.5,0) -- (0,4);
\draw[postaction={decorate}] (4.5,0) -- (6,4);
\end{scope}
\draw[] (3,4) to [controls=+(45:1) and +(135:1)] (6,4);
\draw[] (0,4) to [controls=+(45:1) and +(135:1)] (3,4);
\draw[] (-1.5,0) to [controls=+(45:0.7) and +(135:0.7)] (0.5,0);
\draw[] (0.5,0) to [controls=+(45:0.7) and +(135:0.7)] (2.5,0);
\draw[] (2.5,0) to [controls=+(45:0.7) and +(135:0.7)] (4.5,0);
\filldraw (0,4) circle (0.3cm);
\draw (-1.5,4) node[scale=0.7]{$x$};
\draw (0.5,5.5) node[scale=0.7]{$g$};
\draw (2,-6) node{$y_g=at^{-1}$};
\draw (-4,0) node[scale=0.7]{$T^{-1}(x)$};
\draw (-3.5,9) node[scale=0.7]{$T^{-1}(x)$};
\draw (3,11) node[scale=0.7]{$T^{-1}(x)$};
\draw (5,-3) node[scale=0.7]{$T(x)$};
\draw (8,7) node[scale=0.7]{$T^{-1}(x)$};
\end{scope}

\begin{scope}[shift={(72,0)}]
\begin{scope}[decoration={
    markings,
    mark=at position 0.55 with {\arrow[scale=2]{>>}}}
    ]
\draw[postaction={decorate}] (4.5,0-1.5) -- (3,4);
\begin{scope}[scale=3/2,shift={(2,4*2/3)}]
\draw[postaction={decorate}] (0,4) -- (0,0);
\end{scope}
\begin{scope}[scale=3/2,shift={(0,4*2/3)}]
\draw[postaction={decorate}] (-1.5,4-1.5) -- (0,0);
\end{scope}
\begin{scope}[scale=3/2,shift={(4,4*2/3)}]
\draw[postaction={decorate}, color=rouge] (0,0) -- (1.5,4-2.5);
\end{scope}
\draw[postaction={decorate}] (-1.5,0) -- (0,4);
\draw[postaction={decorate}] (4.5,0) -- (6,4);
\end{scope}
\draw[] (3,4) to [controls=+(45:1) and +(135:1)] (6,4);
\draw[] (0,4) to [controls=+(45:1) and +(135:1)] (3,4);
\draw[] (-1.5,0) to [controls=+(45:0.7) and +(135:0.7)] (0.5,0);
\draw[] (0.5,0) to [controls=+(45:0.7) and +(135:0.7)] (2.5,0);
\draw[] (2.5,0) to [controls=+(45:0.7) and +(135:0.7)] (4.5,0);
\filldraw (0,4) circle (0.3cm);
\draw (-1.5,4) node[scale=0.7]{$x$};
\draw (0.5,5.5) node[scale=0.7]{$g$};
\draw (2,-6) node{$y_g=a^2t^{-1}$};
\draw (-4,0) node[scale=0.7]{$T^{-1}(x)$};
\draw (-3.5,9) node[scale=0.7]{$T^{-1}(x)$};
\draw (3,11) node[scale=0.7]{$T^{-1}(x)$};
\draw (5,-3) node[scale=0.7]{$T^{-1}(x)$};
\draw (8,7) node[scale=0.7]{$T(x)$};
\end{scope}

\end{tikzpicture}    
    \end{bigcenter}
    \caption{The flow configuration drives the choice for computing $T$ or $T^{-1}$ in the different sheets of $BS(2,3)$.}
    \label{fig:flow_patterns_BS_compute_T}
\end{figure}
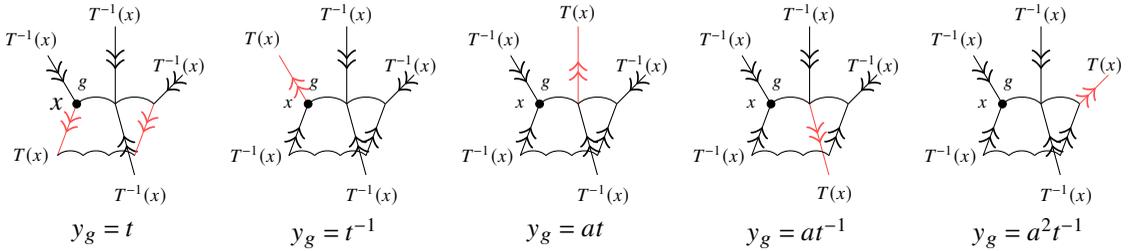
\setlength{\abovecaptionskip}{5pt plus 3pt minus 2pt} 

Combining the formulas from~(\ref{eq:tilecolors}) and the patterns from Figure~\ref{fig:flow_patterns_BS} we can picture Wang tiles from $\tau$ as below.

\begin{figure}[!ht]
    \begin{bigcenter}
    \begin{tikzpicture}[scale=0.6]

\draw[thick] (0,0) -- (0,4)-- (0,4) to [controls=+(45:1) and +(135:1)] (3,4) -- (3,4) to [controls=+(45:1) and +(135:1)] (6,4) -- (6,4) -- (6,0) -- (6,0) to [controls=+(135:0.7) and +(45:0.7)] (4,0) -- (4,0) to [controls=+(135:0.7) and +(45:0.7)] (2,0) -- (2,0) to [controls=+(135:0.7) and +(45:0.7)] (0,0) -- cycle ;

\begin{scope}[decoration={
    markings,
    mark=at position 0.55 with {\arrow[scale=2]{>>}}}
    ]
\draw[color=red,postaction={decorate}](6,4) -- (6,0);
\draw[color=red,postaction={decorate}](0,4) -- (0,0);
\end{scope}

\draw (1.5,5) node[scale=0.9]{$t_1(x,g)$};
\draw (4.5,5) node[scale=0.9]{$x_2(x,g)$};
\draw (-1.25,2) node[scale=0.9]{$\ell(T,x,g)$};
\draw (7.25,2) node[scale=0.9]{$r(T,x,g)$};
\draw (1,0.8) node[scale=0.6]{$b_1(T,x,g)$};
\draw (3,0.8) node[scale=0.6]{$b_2(T,x,g)$};
\draw (5,0.8) node[scale=0.6]{$b_3(T,x,g)$};

\begin{scope}[shift={(12,0)}]
\draw[thick] (0,0) -- (0,4)-- (0,4) to [controls=+(45:1) and +(135:1)] (3,4) -- (3,4) to [controls=+(45:1) and +(135:1)] (6,4) -- (6,4) -- (6,0) -- (6,0) to [controls=+(135:0.7) and +(45:0.7)] (4,0) -- (4,0) to [controls=+(135:0.7) and +(45:0.7)] (2,0) -- (2,0) to [controls=+(135:0.7) and +(45:0.7)] (0,0) -- cycle ;

\begin{scope}[decoration={
    markings,
    mark=at position 0.55 with {\arrow[scale=2]{>>}}}
    ]
\draw[color=red,postaction={decorate}](6,0) -- (6,4);
\draw[color=red,postaction={decorate}](0,0) -- (0,4);
\end{scope}

\draw (1.5,5) node[scale=0.8]{$t_1(x,g)$};
\draw (4.5,5) node[scale=0.8]{$t_2(x,g)$};
\draw (-1.5,2) node[scale=0.8]{ $\ell(T^{-1},x,g)$};
\draw (7.5,2) node[scale=0.8]{$r(T^{-1},x,g)$};
\draw (1,0.8) node[scale=0.5]{$b_1(T^{-1},x,g)$};
\draw (3,0.8) node[scale=0.5]{$b_2(T^{-1},x,g)$};
\draw (5,0.8) node[scale=0.5]{$b_3(T^{-1},x,g)$};
\end{scope}

\end{tikzpicture}    
    \end{bigcenter}
    \caption{Wang tileset $\tau$ that computes $T$ along a flow configuration for $BS(2,3)$.}
    \label{fig:wang_compute_BS23}
\end{figure}
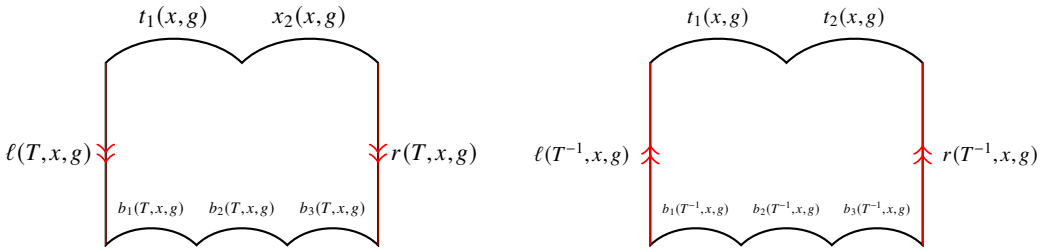
\setlength{\abovecaptionskip}{5pt plus 3pt minus 2pt} 

\begin{prop}\label{prop:tiles_compute}
The tile pictured to the left of Figure~\ref{fig:wangtile_BS23} computes $T$ along $t$, and the tile to the right computes $T^{-1}$ along $t$ --or $T$ against $t$.
\end{prop}

\begin{proof}
The tiles are simplified version of the tiles in~\cite{AubrunKari2021}, since we have a one-dimensional function $T$ or $T^{-1}$ instead of a two-dimensional one, and our functions are linear and not affine. The calculations are left to the reader: the main idea is that terms on top and bottom telescope and the left and right carries precisely compensate the remaining terms.
\end{proof}

\begin{rem}
The proof of Proposition~\ref{prop:tiles_compute} does not depend on the choice for the function~$\lambda$.
\end{rem}

\begin{prop}\label{prop:Y_T_non_empty}
There exists a configuration in $Y_T$.
\end{prop}

\begin{proof}
The proof follows the proof of Lemma~9 from~\cite{AubrunKari2021}, since the function $T$ we have chosen has immortal points. Fix a flow configuration $y\in Y_\text{flow}$ and choose $x$ an immortal point for $T$. For every $g\in BS(2,3)$ we put the tile $\tau(f,T^{\beta_y(g)}(x),g)$ in $g$, where $f=T$ if $y_g=t$ and $f=T^{-1}$ otherwise. This defines a configuration $z$ in $\tau^{BS(2,3)}$. It remains to check that it is indeed in the SFT $Y_T$. We need to check that the three matching rules conditions on page~\pageref{eq:matching_rules} are satisfied. 

\begin{enumerate}
\item $z_{g}(r) = z_{g\cdot a^2}(\ell)$? We distinguish two cases, depending on whether $y_g=t$ or not. If this is the case, then $z_{g}(r)$ is $r(T,T^{\beta_y(g)}(x),g)$:
\[
z_{g}(r)=\frac{1}{2}T\left(\lfloor \left(2\lambda(g)+2\right)T^{\beta_y(g)}(x) \rfloor\right) - \frac{1}{3}\lfloor \left(3\lambda(g)+3\right)T^{\beta_y(g)+1}(x) \rfloor.
\]
In this case, by the allowed patterns of Figure~\ref{fig:flow_patterns_BS}, we also have that $y_{g\cdot a^2}=t$. Thus $z_{g\cdot a^2}(\ell)$ is $\ell(T,T^{\beta_y(g\cdot a^2)}(x),g\cdot a^2)$ and since $\beta_y(g\cdot a^2)=\beta_y(g)$ we get:
\[
z_{g\cdot a^2}(l)=\frac{1}{2}T\left(\lfloor 2\lambda(g\cdot a^2)T^{\beta_y(g)}(x) \rfloor\right) - \frac{1}{3}\lfloor 3\lambda(g\cdot a^2)T^{\beta_y(g)+1}(x) \rfloor.
\]
It suffices to use the fact that $\lambda(g\cdot a^2)=\lambda(g)+1$ to conclude.

\medskip

In the second case, $y_g\neq t$,  $z_{g}(r)$ is $r(T^{-1},T^{\beta_y(g)}(x),g)$ and the allowed patterns of Figure~\ref{fig:flow_patterns_BS} impose that $y_{g\cdot a^2}\neq t$. Thus $z_{g\cdot a^2}(\ell)$ is equal to $\ell(T^{-1},T^{\beta_y(g)\cdot a^2}(x),g\cdot a^2)$. The equalities $\lambda(g\cdot a^2)=\lambda(g)+1$ and $\beta_y(g\cdot a^2)=\beta_y(g)$ give that $z_{g}(r) = z_{g\cdot a^2}(\ell)$.

\item $z_{g}(b_{i+1}) = z_{g\cdot ta^{i}}(t_1)$ for $i=0,1,2$?

 If $y_g=t$, then 
 \begin{align*}
  z_{g}(b_{i+1}) &= b_{i+1}(T,T^{\beta_y(g)}(x),g)\\
  &= \lfloor\left(3\lambda(g)+i+1 \right)T^{\beta_y(g)+1}(x) \rfloor - \lfloor \left(3\lambda(g)+i \right)T^{\beta_y(g)+1}(x) \rfloor.
 \end{align*}

On the other hand,
 \begin{align*}
  z_{g\cdot ta^{i}}(t_1) &= t_1(T^{\beta_y(g\cdot ta^{i})}(x),g\cdot ta^{i})\\
  &= \lfloor \left(2\lambda(g\cdot ta^{i})+1 \right)T^{\beta_y(g\cdot ta^{i})}(x) \rfloor - \lfloor \left(2\lambda(g\cdot ta^{i}) \right)T^{\beta_y(g\cdot ta^{i})}(x) \rfloor.
 \end{align*}
 Using results from Proposition~\ref{prop:lambda_useful_eq} we get that
 \begin{align*}
  z_{g\cdot ta^{i}}(t_1) &= \lfloor \left(2\left(\frac{3}{2}\lambda(g)+ \frac{i}{2}\right)+1 \right)T^{\beta_y(g)+1}(x) \rfloor - \lfloor \left(2\left(\frac{3}{2}\lambda(g)+ \frac{i}{2}\right) \right)T^{\beta_y(g)+1}(x) \rfloor\\
  &= \lfloor \left(3\lambda(g)+i+1 \right)T^{\beta_y(g)+1}(x) \rfloor - \lfloor \left(3\lambda(g)+i \right)T^{\beta_y(g)+1}(x) \rfloor\\
  &= z_{g}(b_{i+1}).
 \end{align*}
 
 If $y_g\neq t$, the calculations are quite similar, except that $T$ is replaced by $T^{-1}$ in the expression of $z_g(b_{i+1})$, which is compensated by the fact that, in that case, $\beta_y(g\cdot ta^{i})=\beta_y(g)-1$.
 
\item $z_{g}(b_{i+1}) = z_{g\cdot ta^{i-1}}(t_2)$ for $i=0,1,2$?

 This part is very similar to what precedes and left to the reader, since $t_2(x,g)$ is just a \emph{shift} of $t_1(x,g)$.
\end{enumerate}

\end{proof}

\begin{prop}\label{prop:Y_T_aperiodic}
The SFT $Y_T$ is strongly aperiodic.
\end{prop}

\begin{proof}
Let $z=(y,\tau)$ be a configuration in $Y_T$ and assume it possesses a period $g\in BS(2,3)$. Thus for every $k\in\Z$ one has that
\[
z_{a^k}=z_{g^{-1}\cdot a^k}
\]
so that the two $\langle a \rangle$-cosets at $1$ and $g^{-1}$ are the same. If we denote by $x$ the real number encoded by $\tau$ on the $\langle a \rangle$-coset of the identity, we get that 
\[
T^{\beta_y(g)}(x)=x.
\]
But the periodicity of $z$ also constraints the flow configuration $y$. Necessarily by Proposition~\ref{prop:flow_follows_period_BS}, if we decompose $g$ into its normal form $g = wa^{p}$, we have that $y$ is characterized by either the infinite word $w^{\N}$ or $(w^{-1})^{\N}$. Without loss of generality we take $W(y) = w^{\N}$, which in particular implies that $\beta_{y}(w) = \beta_y(g)=|g|_t+|g|_{t^{-1}}$. Hence we can rewrite $T^{\beta_y(g)}(x)=x$ as 
\[
T^{|g|_t+|g|_{t^{-1}}}(x)=x
\]
which in turn implies, by the aperiodicity of $T$, that $|g|_t+|g|_{t^{-1}}=0$. Because the two terms are positive they are necessarily zero. The period $g$ is therefore a power of $a$ that we denote $a^{-N}$ for some $N\in\Z$. We now know that for every group element $h\in BS(2,3)$
\[
z_{h}=z_{a^N\cdot h},
\]
so that each $\langle a \rangle$-coset in the configuration $y$ wears a $N$-periodic bi-infinite word. Since there are only finitely many possible words of length $N$, by following the flow component of $y$, there must exist two distinct integers $k,k'$ such that $T^k(x)=T^{k'}(x)$. Again the aperiodicity of $T$ implies that $k=k'$, which contradicts our initial assumption. We conclude that $z$ has no period.
\end{proof}

Combining Proposition~\ref{prop:Y_T_non_empty} and Proposition~\ref{prop:Y_T_aperiodic} gives the existence of a strongly aperiodic SFT on $BS(2,3)$. Since all non-residually finite Baumslag-Solitar groups are finitely presented, torsion free and quasi-isometric between them, Theorem~\ref{theorem:Cohen_QI} of~\cite{cohen2017large} applies and we conclude that all the $BS(m,n)$ with $m,n>1$ and $m\neq n$ admit strongly aperiodic SFTs. 

\begin{theorem}\label{thm:BSmn_strongly_aperiodic_SFT}
Non-residually finite Baumslag-Solitar groups $BS(m,n)$ with $m, n > 1$ and $m\neq n$ admit strongly aperiodic SFTs.
\end{theorem}

\begin{cor}\label{cor:GBS_strongly_aperiodic_SFT}
All non-$\Z$ GBS groups admit a strongly aperiodic SFT.
\end{cor}

\subsection{Consequences}

Through the machinery provided by Theorem \ref{theorem:Cohen_QI}, we can push the result to a broader class of groups, namely those obtained as the fundamental group of a graph of virtual $\Z$'s. This structure is the same as in Definition \ref{def:gbs} but all vertex groups are virtually $\Z$ instead of just $\Z$.

\begin{theorem}[(\cite{mosher2003quasi})]
A group $G$ is quasi-isometric to a GBS group if and only if it is the fundamental group of a graph of virtual $\Z$'s.
\end{theorem}

This way, Corollary \ref{cor:GBS_strongly_aperiodic_SFT} implies the following result.

\begin{cor}
Let $G$ be the fundamental group of a graph of virtual $\Z$'s. If $G$ is not virtually $\Z$ it admits a strongly aperiodic SFT.
\end{cor}

\section*{Conclusion}

We have shown the existence of strongly aperiodic SFTs for all GBS groups and even graphs of virtual $\Z$'s that are not virtually $\Z$. Moreover, for non-$\Z$ unimodular GBS groups and residually finite Baumslag-Solitar groups, the strongly aperiodic SFT can also be chosen minimal.

\medskip

A question that remains open is the existence of a minimal strongly aperiodic SFT on non-residually finite Baumslag-Solitar groups $BS(m,n)$. Note that the existence of such an SFT would imply the same for all non-virtually $\Z$ GBS groups. Another related question concerns strongly aperiodic SFTs for Artin groups: can the path-folding technique be adapted to this class of groups?



\begin{Backmatter}


\bibliographystyle{apalike}
\bibliography{biblio.bib}

\end{Backmatter}

\end{document}